\documentclass[acmsmall,nonacm,screen]{acmart}

\AtBeginDocument{%
  }

\usepackage{wrapfig}
\usepackage{graphicx}
\usepackage{float}
\usepackage{qcircuit}
\usepackage{pifont}
\usepackage{tikz}
\usepackage{algorithm}
\usepackage{algorithmicx}
\usepackage[noend]{algpseudocode}
\usepackage{subcaption}
\usepackage{listings}
\usepackage{xcolor}
\usepackage{makecell}
\usepackage{mathpartir}
\usepackage{booktabs}
\usepackage{multirow}
\begin{document}

\title{QisMC: A Model Checker for QISKIT Program Debugging}

\author{Aochu Dai}
\email{dac22@mails.tsinghua.edu.cn}
\orcid{0009-0005-9935-8823}
\affiliation{%
  \institution{Department of Computer Science and Technology, Tsinghua University}
  \city{Beijing}
  \country{China}
}

\author{Mingsheng Ying}
\affiliation{%
  \institution{Centre for Quantum Software and Information, University of Technology Sydney}
  \city{Sydney}
  \country{Australia}}
\email{Mingsheng.Ying@uts.edu.au}

\renewcommand{\shortauthors}{A. Dai and M. Ying}

\begin{abstract}
    We present QisMC, the first quantum model checker dedicated to debugging Qiskit programs. On the theoretical side, we introduce the notion of quantum-classical transition system and a quantum computation tree logic (qCTL) grounded in Birkhoff-von Neumann logic for modeling the behaviors and specifying the properties of Qiskit programs, respectively. On the implementation side, QisMC provides an end-to-end framework encompassing transition system generation, logical formula representation, and model checking algorithms, and it efficiently performs image computation based on decision diagrams. These methods and design principles give QisMC advantages over previous Qiskit program debuggers, including a unified property specification language, a fully automated and exhaustive verification process, and the ability to generate counterexamples. Extensive illustrative examples and benchmark evaluations demonstrate the practicality, efficiency, and scalability of QisMC in verifying realistic quantum programs.
\end{abstract}




\maketitle

\section{Introduction}
Recent advances in quantum computing~\cite{acharya_quantum_2025,sales2025experimental} are driving quantum programs beyond the conventional model of static circuits consisting of fixed sequences of unitary operations, toward richer forms of quantum-classical interaction. Mid-circuit measurements allow a program to acquire classical information during execution and use it to dynamically determine subsequent quantum operations through classical feedforward, conditional operations, and control flow. Programs exhibiting such behaviors are commonly referred to as dynamic quantum circuits/programs~\cite{corcoles2021exploiting,sequential}. Such dynamic behaviors have already emerged in several important quantum computing tasks. For example, quantum error correction repeatedly performs syndrome measurements and uses their outcomes for classical decoding and subsequent operations. Similar quantum-classical interactions are also prevalent in protocols such as quantum teleportation~\cite{pirandola2015advances}, LOCC~\cite{chitambar2014everything}, and distributed quantum computing~\cite{caleffi2024distributed}. Meanwhile, mainstream quantum programming frameworks such as Qiskit have begun to provide direct programming support for mid-circuit measurements, classical variables, and structured control flow, together with evolving support for their transpilation and execution on quantum hardware~\cite{qiskit2024}. Consequently, the behavior of modern quantum programs is determined not only by operations on quantum states, but also by classical information obtained from measurements and the different execution paths selected accordingly.

However, such dynamic quantum-classical interactions also pose new challenges for the verification and debugging of quantum programs. \textbf{First, specifying the intended behavior of a program is itself nontrivial.} For static quantum circuits, correctness can often be characterized by input-output relations or properties of quantum states at particular program locations. For programs involving branches and loops, however, correctness may additionally concern how quantum states evolve along execution paths—for example, whether a property is always preserved, whether a particular quantum state is eventually reachable, or whether one property continues to hold until another event occurs. \textbf{Second, when a program violates its intended property, merely reporting verification failure provides limited assistance for debugging.} Developers also need to understand along which execution the property is violated and which program locations and operations contribute to the violation. Existing techniques and tools based on quantum program testing~\cite{mutationtestQiskit,qucheck,Shahin2020Property}, runtime assertions~\cite{Li2020Projection, Liu2020Runtime, Liu2021systematic}, and static analysis~\cite{Huang2019statistical,LintQ} have demonstrated the ability to detect various classes of quantum program bugs~\cite{rovara2024frameworkdebuggingquantumprograms}. However, they typically focus on particular program points, bounded executions, or specific classes of program properties, and do not simultaneously provide exhaustive reasoning over complex control flow, temporal specification, and diagnostic information that explains the execution leading to a property violation.

These two requirements align closely with the core strengths of model checking in classical hardware and software verification. Model checking allows developers to specify expected system behaviors over different execution paths using temporal logic and automatically determine whether a given property holds through systematic exploration of the state space~\cite{baier2008principles}. More importantly, when a property is violated, a model checker can typically generate a counterexample execution that leads to the violation, allowing developers to trace a concrete execution path to locate and understand the error. This diagnostic capability has long made model checking an effective technique for debugging, as emphasized by Clarke, Emerson and Sifakis in their 2007 ACM A.M. Turing Award Lecture~\cite{Clarke2009lecture}. Model checking therefore provides a natural framework for addressing the two challenges in quantum programming discussed above: temporal logic specifies correctness requirements across program locations and execution paths, while counterexamples turn verification failures into concrete diagnostic information for debugging.

Prior work has extended model checking to quantum systems and developed several quantum temporal logics and verification frameworks~\cite{feng2013model,do2024symbolic,Gay2008QMC,ying2021model}. However, as discussed in Section~\ref{RelatedWork}, existing approaches provide limited support for integrating quantum-state properties, realistic program control flow, and effective debugging within a unified workflow. To address these limitations, we present QisMC, a model checker for debugging Qiskit programs based on model-checking techniques and Birkhoff-von Neumann quantum logic. QisMC extends the classical model-checking workflow of temporal specification, exhaustive verification, and counterexample-guided debugging to quantum programs that combine quantum operations with classical control flow. To this end, we define a new quantum temporal logic, qCTL, based on Birkhoff-von Neumann quantum logic, and develop an end-to-end workflow encompassing automatic model construction from Qiskit programs, property specification and verification, and counterexample-guided debugging. We next outline the two key design ideas underlying QisMC.  
\begin{itemize}\item \textbf{Specification}: 
For expressing more sophisticated properties of Qiskit programs, 
QisMC employs a carefully designed classical-quantum computation tree logic (qCTL) that incorporates  Boolean-valued classical formulas with quantum atomic propositions defined over Hilbert subspaces. Recall that in Birkhoff-von Neumann logic~\cite{birkhoff1975logic,Chiara2004Reasoning}, a quantum proposition corresponds to a Hilbert subspace of the system, and a quantum state $\rho$ satisfies a proposition $P$ if $\mathrm{tr}(P\rho) = 1$, where $\mathrm{tr}$ denotes the trace operation. This subspace-based formalism naturally captures many properties central to quantum computing, e.g. state equivalence, termination, reachability, and valid codeword spaces in quantum error correction. Indeed, such properties have already been used in prior work on quantum program debugging and testing~\cite{qucheck,huang2025,Li2020Projection}. In this work, they are extended to the \textit{temporal}  setting. As will be illustrated from our motivating example in Subsection \ref{MotExample}, qCTL formulas can effectively express practical temporal properties of quantum programs. 

\item \textbf{Debugging}: QisMC provides an end-to-end workflow for counterexample-guided debugging of Qiskit programs. Through its Python interface, users can construct the transition system, define quantum and classical propositions, and associate properties with program locations through inline annotations. QisMC then employs the bidirectional quantum fixed-point iterations introduced in Section~\ref{sec:bi-fixed} to determine the satisfaction of quantum propositions at different program locations. Based on these results, the qCTL model-checking problem is reduced to a conventional CTL problem and automatically encoded for NuSMV~\cite{nusmv}, which performs exhaustive model checking and generates a counterexample when the specification is violated. Finally, QisMC maps the counterexample from the transition-system level back to the original Qiskit program, producing a source-level execution path that leads to the property violation and provides concrete diagnostic information for debugging.  
\end{itemize}

\subsection{A Motivating Example}\label{MotExample}
To illustrate the application scenarios and verification goals of QisMC, and to introduce the theoretical and tool foundations developed in this paper, we present a small yet complete Qiskit program debugging process as a running example. The example includes a simple control-flow construct---a while loop---that is rarely covered by existing Qiskit verification tools mentioned above.
\begin{example}[RUS program in Qiskit]\label{eg:rus1}
    RUS (Repeat-Until-Success) circuits have been used to construct some special single-qubit unitary gates. Suppose a quantum programmer wishes to implement the gate $\frac{I+i\sqrt2 X}{\sqrt{3}}$ following the construction in~\cite{rus14}. He writes the Qiskit code shown in Fig.~\ref{fig:rus_buggy_prog}, which produces the circuit in Fig.~\ref{fig:rus_buggy_circ} using the built-in circuit visualization tool.

    The program employs a \emph{while-loop} construct together with an ancilla qubit $q_2$, whose role is to initialize a classical register to the value 1 through measurement. After initializing the quantum state of $q_0$, the program enters the while loop. At the end of each iteration, $q_0$ is measured; if the outcome is 0, the loop terminates, thereby implementing the intended quantum gate.

    However, due to an oversight by the programmer, a bug appears in the loop body. When the  execution returns to the beginning of the loop, $q_0$ should be reset. Otherwise, the circuit no longer applies the intended operation repeatedly to the same quantum state.
\end{example}

\begin{figure}[htbp]
  \centering
  \begin{minipage}[b]{0.4\textwidth}
    \centering
    \includegraphics[width=\textwidth]{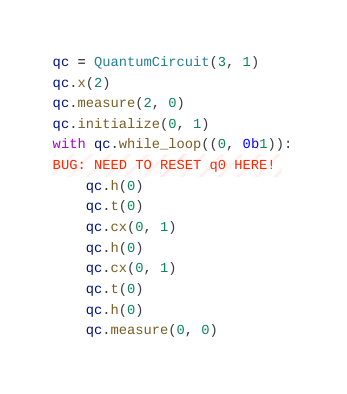}
    \subcaption{RUS program with bugs.}
    \label{fig:rus_buggy_prog}
  \end{minipage}
  \hfill
  \begin{minipage}[b]{0.5\textwidth}
    \centering
    \includegraphics[width=\textwidth]{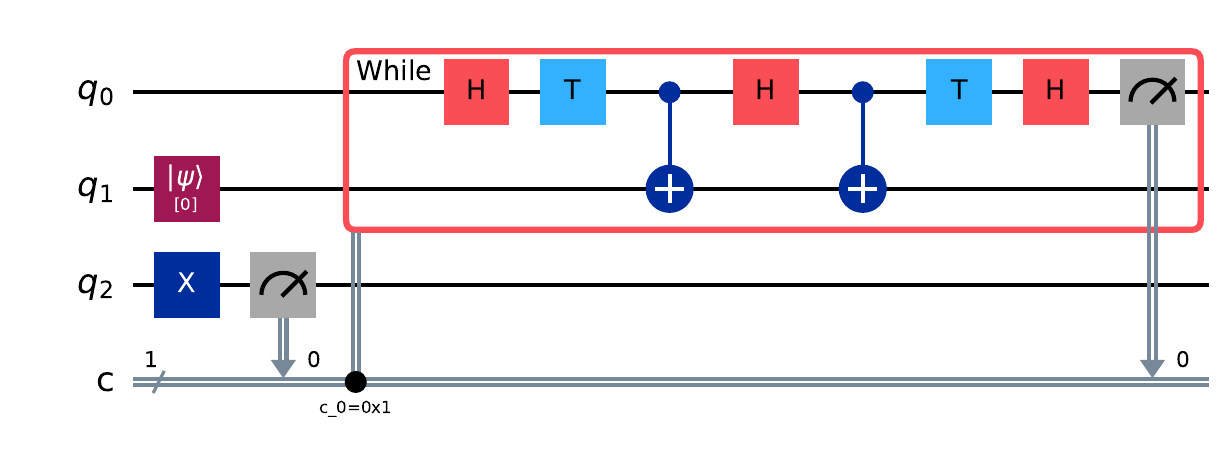}
    \subcaption{The quantum circuit of (a) represented by Qiskit.}
    \label{fig:rus_buggy_circ}
    \vspace{0.5em}
    \includegraphics[width=\textwidth]{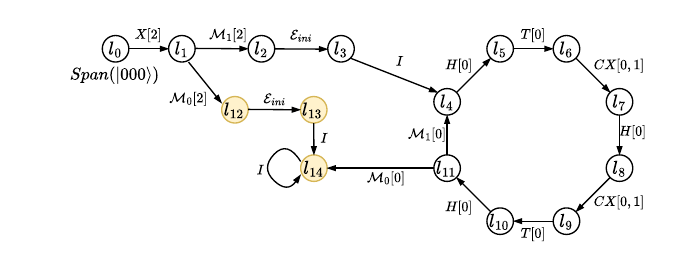}
    \subcaption{The transition system of the program, introduced in Section~\ref{sec3}}
    \label{fig:rus_buggy_qcts}
  \end{minipage}
  \caption{A buggy repeat-until-success program, circuit, and transition system to implement the $\frac{I+i\sqrt2 X}{\sqrt{3}}$ gate.}
  \Description{Repeat-until-success circuits aiming to implement the $\frac{I+i\sqrt2 X}{\sqrt{3}}$ gate.}
  \label{rusdebugging}
\end{figure}
In Qiskit programming, bugs may arise at various stages of program design~\cite{huangQDB2019,Leite2025testing}. Typical examples include incorrectly applied gates or measurements (e.g., due to a wrong operation or target qubit), faulty classical control flow, and improper initialization. Therefore, checking functional and temporal properties of quantum programs is essential.

In classical model checking, well-defined linear or branching-time temporal logics can express a rich class of properties, such as safety, liveness, fairness, and invariance. For Example~\ref{eg:rus1}, one may naturally consider properties such as \textbf{liveness---whether the program eventually exits the loop}---and \textbf{safety---whether the program terminates in the intended target quantum state}. In Section~\ref{sec:qctl-def}, we introduce a new temporal logic framework and employ a safety property to debug Example~\ref{eg:rus1}.

The \textbf{existing tools} for Qiskit program debugging (e.g. ~\cite{qucheck,rovara2024frameworkdebuggingquantumprograms, Li2020Projection}) typically capture the impact of bugs through assertions or the specification of program properties. However, conventional approaches provide limited support for expressing and checking how such assertions should hold over temporal executions, or require relatively involved manual designs to achieve similar reasoning. Moreover, they generally provide limited diagnostic information beyond whether a given assertion or property is satisfied. \textbf{In contrast}, QisMC directly supports the temporal property required in this example and, upon detecting a violation, produces a counterexample execution that traverses the initial location, the location where the bug is introduced (the entry of the \texttt{while} loop), and the location where its effect becomes observable (the loop exit). The counterexample can be further enriched with diagnostic information along the execution, such as the evolution of the quantum state, helping programmers understand how the bug propagates through the program and eventually leads to the property violation.

\subsection{Contributions and Outline}
In this work, we design and implement the complete model checking process of QisMC. The overall structure of QisMC's workflow is visualized in Fig.~\ref{fig:QisMCworkflow}. The areas highlighted in green there indicate our innovations. More explicitly, our main contributions can be summarized as follows:
\begin{itemize}
    \item \textbf{Advances in theoretical foundations.}
    We introduce a quantum computation tree logic (qCTL) for reasoning about quantum-classical transition systems. 
    The logic naturally supports classical control flow and variables, while enabling the verification of programs with infinitely many runtime quantum states through subspace-based quantum propositions. 
    Based on a bidirectional fixed-point iteration algorithm, we establish a reduction from qCTL model checking to a classical CTL model checking problem, which allows modern classical model checking techniques to be directly applied to the debugging of Qiskit programs.

    \item \textbf{A complete workflow for debugging Qiskit programs.}
    QisMC is the first quantum model checker dedicated to debugging Qiskit programs. 
    It supports a wide range of classical control-flow constructs in Qiskit, including classical variables, feedforward operations, and measurements. 
    Moreover, QisMC provides several capabilities that are absent from existing Qiskit property checkers~\cite{qucheck}, including flexible specification of temporal properties, automatic generation of counterexamples, and support for quantum while-loop constructs.

    \item \textbf{Symbolic reasoning for quantum propositions.}  
    We develop a practical framework for manipulating quantum propositions, including operations such as intersections, union, and complement of subspaces that arise in qCTL model checking. To support scalable implementation, QisMC employs a decision-diagram-based symbolic representation for quantum states and operations (implemented with CFLOBDDs~\cite{cflobdd} in the current version). To the best of our knowledge, this is the first work that explores decision-diagram-based symbolic representations in the context of quantum model checking.
\end{itemize}

The remainder of this paper is organized as follows. Section~\ref{qiskit-program} introduces the necessary background on Qiskit programming. Sections~\ref{sec3} and \ref{sec4} present the main theoretical foundations of QisMC---qCTL model checking. Sections~\ref{design} and \ref{experiments} describe the design, implementation, and experimental evaluation of QisMC. Related work is  discussed in Section \ref{RelatedWork}. Section~\ref{conclusion} concludes the paper.

\begin{figure}
    \centering
    \includegraphics[width=0.9\linewidth]{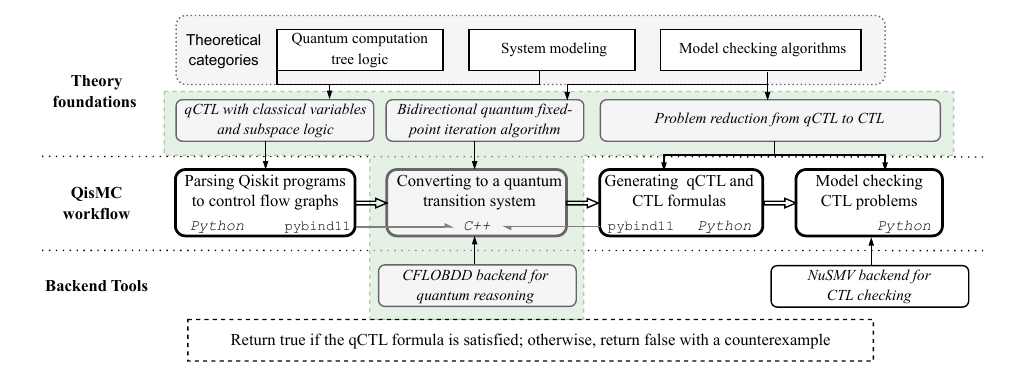}
    \caption{Overall structure of QisMC's workflow. 
    }
    \label{fig:QisMCworkflow}
    \Description{Overall structure of QisMC's theory, theory contributions, and workflow}
\end{figure}

\section{Related Work}\label{RelatedWork} 
In this section, we present several existing techniques and tools for quantum program debugging and testing along with quantum model checking. Some of the mentioned research is also surveyed in~\cite{Leite2025testing}. These approaches are categorized according to their respective domains, and the theoretical foundations of some representative tools are briefly discussed.

\subsection{Debugging and Testing of Quantum Programs}

\textit{Property-based testing} provides a structural framework for testing properties of quantum programs. Typical workflows—such as those in~\cite{qucheck,Shahin2020Property}—consist of property specification, test case generation, test execution, and statistical analysis. Similarly, MQT-debugger~\cite{rovara2024frameworkdebuggingquantumprograms} offers a user-friendly debugging framework that enables the manual insertion of assertions into programs, supporting the verification of entanglement, superposition, and equality properties.

\textit{Assertion in quantum program debugging} has received growing attention. Owing to the inherent complexities of quantum computation—such as its probabilistic nature and entanglement—the representation of properties and assertions is non-trivial. To address this challenge, a variety of assertion mechanisms have been investigated, including statistical assertions\cite{Huang2019statistical}, projection-based assertions\cite{Ying2022Assertion,Li2020Projection}, and dynamic or approximate assertions~\cite{Liu2021systematic,Liu2020Runtime}.

\textit{Mutation testing} evaluates the adequacy of quantum test suites by introducing mutation operators into programs to generate faulty variants. Several mutation testing tools have been developed for quantum programs, including~\cite{mutationtestQiskit,Daniel2022mutation,Muskit2021}.



\textit{Quantum Hoare logic.}
For quantum programs, there existed a set of mature Floyd-Hoare-style theories based on assertion logic, called quantum Hoare logic~\cite{qhl,fengClassical,zhou2019applied}, in which the mathematical properties of quantum states are specified by an assertion language with a certain limit of expressiveness. Practical quantum Hoare logic toolkits have been implemented in interactive proof assistants~\cite{coqq2023,Rand_2018}. Recently, research on quantum program verification has increasingly focused on practical and scalable problems, such as quantum error correction and fault-tolerant quantum computation~\cite{huang2025}.

\textit{Other formal verification methods} including~\cite{Chen2023Automata}, which supports quantum program verification based on automata-based framework, and~\cite{Chen2025,fang2024symbolic}, which are built on symbolic execution of quantum programs. We refer readers who are interested in quantum program debugging to the monograph~\cite{YINGFOUNDATION} or the survey~\cite{Leite2025testing} for further reading.

\subsection{Model Checking Quantum Systems}

The theory of quantum model checking has undergone nearly two decades of development. In early work~\cite{formal2002,Gay2006Probabilistic}, classical methods and tools of probabilistic model checking were first applied in the verification of quantum protocols. Later in the work of 2006~\cite{weakly2006}, an exogenous quantum propositional logic was introduced and subsequently extended to both branching-time and linear-time temporal logic model checking frameworks~\cite{temporal2009,Baltazar2008QCTL,towards2007}. To the best of our knowledge, this constitutes the earliest formulation of QLTL and QCTL.

Subsequent research focused on modeling quantum systems and temporal properties in a more theoretical and precise way. Quantum Markov chain for model checking was proposed in~\cite{Yu2012Reachability}, modeling a quantum program with loops. Based on this model, a range of techniques on different model checking problems were proposed, including reachability analysis~\cite{Yu2012Reachability,dai2024qreach}, reachability probability~\cite{reachabilityProb}, and LTL model checking~\cite{guan2024measurement}. An alternative but expressively equivalent model, called super-operator valued Markov chain, was proposed in~\cite{feng2013model}, which allowed verification of important probabilistic properties while maintaining a finite, discrete classical state space. Based on this model, a series of variations of temporal logic and applications of program verification have also been developed in~\cite{feng2017model,Xu2022qctl,feng2015qpmc}.

QMC~\cite{Gay2008QMC} was presented by Simon Gay et al., based on the early exogenous quantum propositional logic~\cite{towards2007, weakly2006, temporal2009, Baltazar2008QCTL}. It takes quantum protocols written in a custom concurrent modeling language and explores process interleavings, nondeterministic choices, and measurement outcomes to construct an execution tree. Quantum states are restricted to the stabilizer formalism in the system.
For specification, QMC adopts QCTL, a branching-time temporal logic obtained by extending CTL with dEQPL~\cite{weakly2006} state formulas. These formulas can characterize quantum states through amplitudes, measurement probabilities, and entanglement properties, while also referring to classical variables and measurement outcomes. However, the tool paper demonstrates only a small quantum coin-flipping example, whose reported properties mainly concern agreement of classical outcomes and a temporal safety condition over classical variables. Moreover, the paper does not report counterexample generation or dedicated debugging support.

QPMC~\cite{feng2015qpmc} is the first probabilistic model checker for quantum Markov chains, implementing the super-operator-valued quantum Markov chain model and the associated model checking algorithm proposed by Feng et al.~\cite{feng2013model}. This approach represents locations as classical states and quantum operations as super-operators labeling transitions, thereby avoiding an explicit continuous quantum state space. It further introduces a super-operator-valued measure over execution paths, allowing probabilistic temporal properties to be characterized uniformly for arbitrary initial quantum states. Built upon IscasMC, QPMC supports quantum programs with classical control flow, including loops and conditional branches. However, its atomic propositions remain purely classical and thus cannot directly express properties of quantum states. Moreover, the tool was demonstrated only on small-scale examples, leaving its scalability to larger quantum programs unclear.

QTC-Maude~\cite{do2024symbolic} provides a symbolic model-checking framework for quantum circuits based on Maude. Quantum operations and states are represented in Dirac notation and the process of quantum reasoning is encoded as rewrite rules. Therefore, the resulting transition system can be verified using Maude’s built-in classical LTL model checker. However, QTC-Maude does not introduce quantum-specific atomic propositions or a quantum temporal logic; quantum-state conditions must instead be manually encoded as Boolean predicates over symbolic states. Moreover, the presented circuit model does not provide control-flow constructs such as loops, and its evaluation is limited to the quantum teleportation case study, leaving its scalability on larger quantum programs unclear.


\section{Qiskit Programming}\label{qiskit-program} 






Qiskit is currently the most widely used open-source toolbox for quantum computing~\cite{qiskit2024}. It provides a full-stack solution for the implementation of quantum algorithms, including qubit mapping, circuit synthesis and optimization, and runtime environment deployment~\cite{qiskit2024}. As its core module, the \texttt{QuantumCircuit} class offers a rich set of functional operations, a lightweight structure, and flexible representations. Qiskit continues to extend the scope of quantum circuits in its updating versions, progressively incorporating support for dynamic quantum circuits, pulse-level operations, and classical expressions. 

\textbf{Syntax of Qiskit programs.} From a programming perspective, the instruction sequence that constitutes a \texttt{QuantumCircuit} can be interpreted as a form of quantum program~\cite{qiskit2024}. Building on this view, QisMC provides model checking capabilities for the commonly used syntactic constructs of \texttt{QuantumCircuit}. Since the proposed approach falls within the domain of formal verification, we present below a formal definition of the class of quantum programs over Qiskit’s \texttt{QuantumCircuit} that are supported by QisMC.
\begin{align*}
    S ::= \textbf{skip}&\mid q:=|0\rangle\mid \bar{q} := vec(\psi)\mid \bar{q} := U[\bar{q}]\mid S_1;S_2\mid x := \textbf{measure } q\\
    &\mid \textbf{if }x=c\textbf{ then } S_1\textbf{ else } S_2\textbf{ fi}\mid \textbf{while } x=c\textbf{ do }S\textbf{ od}
\end{align*}
Compared with the definitions of quantum programs with classical variables in prior work~\cite{fang2024symbolic,fengClassical}, the adopted program has the following difference: (1) We support the \texttt{initialize} construct in Qiskit, which allows for the initialization of an arbitrary quantum state, where $vec(\psi)$ denotes a vector over the complex field $\mathbb{C}$; (2) In the conditions of \texttt{if} and \texttt{while} statements, only the form $x=c$ is supported, where $x$ is a classical variable and $c$ is a constant of the type of $x$. This restriction arises from the limitations of Qiskit version. For the benchmark algorithms considered in this work, the required classical Boolean expressions can be conveniently simulated through the use of ancilla qubits together with additional measurements and control operations.


In Qiskit, each of the syntactic rules corresponds to a construct in the \texttt{QuantumCircuit} class. Quantum variables and (Boolean) classical variables are represented by \texttt{Qubit} and \texttt{Clbit}, respectively. Qiskit further supports the packaging of variables. For example, a \texttt{ClassicalRegister} instance consists of multiple \texttt{Clbit} objects, and can be used as a whole to serve as the conditional variable in an \texttt{if} statement. In Table~\ref{tab:qiskitsyntax}, we present some examples of the formal syntax of Qiskit programs alongside their corresponding Qiskit constructs, which are supported in QisMC.
{\begin{table*}
\footnotesize
  \caption{Syntax and constructs of Qiskit programs}
  \label{tab:qiskitsyntax}
  \begin{tabular}{cc}
    \toprule
    Formal syntax & Qiskit constructs\\
    \midrule
    $q:=|0\rangle, \bar{q} = vec(\psi)$ & \makecell{\texttt{qc.reset(q)}\\ \texttt{qc.initialize([a,b], [q])}}\\
    \midrule
    $\bar{q}:= U[\bar{q}]$ &  \makecell{\texttt{qc.h(q)}\quad\texttt{qc.cnot(q1,q2)}\\ \texttt{qc.append(Pauli(`-iX'), [q])}} \\
    \midrule
    $x := \textbf{measure } q$ & \texttt{qc.measure(q,x)}\\
    \midrule
    $\textbf{if }x=c\textbf{ then } S_1\textbf{ else } S_2$ & \makecell{\texttt{with if\_test((x,c)) as else\_:}...\\ \texttt{with else\_:}...}\\
    \midrule
    $\textbf{while } x=c\textbf{ do }S\textbf{ od}$ & {\texttt{with qc.while\_loop((x, c)):}...}\\
  \bottomrule
\end{tabular}
\end{table*}}

\textbf{Operational semantics} describe the execution sequences of computation steps of a Qiskit program~\cite{YINGFOUNDATION}. It also serves as the theoretical foundation for the analysis and verification of program behavior. Following prior work~\cite{fang2024symbolic,huang2025,fengClassical,ying2025practicalquantumhoarelogic} on the organization of classical and quantum states, we let $\sigma$ denote the values of the classical registers and $\rho$ the quantum states in the quantum registers. We define a triple $\langle S, \sigma, \rho\rangle$ as a quantum configuration. Then the operational semantics of Qiskit programs is defined in Fig.~\ref{fig:opsem} as a transition relation between configurations. 
{\begin{figure}[htbp]
  \centering
  \scriptsize
  {\setlength{\belowdisplayskip}{0pt}
  \setlength{\abovedisplayskip}{0pt}
  \begin{mathpar}
    \langle \textbf{skip}, \sigma, \rho \rangle \to \langle \downarrow, \sigma, \rho \rangle \quad\ \  
    
    \quad \langle \bar{q} := U[\bar{q}], \sigma, \rho \rangle \to \langle \downarrow, \sigma, U\rho U^\dagger\rangle \\

    \inferrule{\langle S_1,\sigma,\rho\rangle\to\langle S_1^\prime,\sigma^\prime,\rho^\prime\rangle}{\langle S_1;\ S_2, \sigma,\rho\rangle\to\langle S_1^\prime;\ S_2,\sigma^\prime,\rho^\prime\rangle}\quad 
    
    \langle q := |0\rangle, \sigma, \rho \rangle 
    \to \langle \downarrow, \sigma,  |0\rangle_{q} \langle 0|\rho|0\rangle_{q} \langle 0| + |0\rangle_{q} \langle 1|\rho|1\rangle_{q} \langle 0|  \rangle \\
    
    \langle \bar{q} := vec(\psi), \sigma, \rho \rangle 
    \to \langle \downarrow, \sigma,  |\psi\rangle_{\bar{q}} \langle 0|\rho|0\rangle_{\bar{q}} \langle \psi| + |\psi\rangle_{\bar{q}} \langle 1|\rho|1\rangle_{\bar{q}} \langle \psi|\rangle
    \and    
    \inferrule{M_0 = |0\rangle_q\langle 0|,\  M_1 = |1\rangle_q\langle 1|}{\langle x := \textbf{measure } q, \sigma, \rho \rangle 
    \to \langle \downarrow, \sigma[i/x], M_i\rho M_i^\dagger/\mathrm{tr}(M_i\rho M_i^\dagger)\rangle}\\
    
    \inferrule{\sigma(x) = c}{\langle \textbf{if }x=c\textbf{ then } S_1\textbf{ else } S_2, \sigma, \rho \rangle 
    \to \langle S_1, \sigma, \rho  \rangle}\quad  
    
    \inferrule{\sigma(x) \neq c}{\langle \textbf{if }x=c\textbf{ then } S_1\textbf{ else } S_2, \sigma, \rho \rangle 
    \to \langle S_2, \sigma, \rho  \rangle}\\

    \inferrule{\sigma(x) = c}{\langle \textbf{while } x=c\ \textbf{do } S\ \textbf{od}, \sigma, \rho \rangle 
  \to \langle S; \textbf{while } x=c\ \textbf{do } S\ \textbf{od}, \sigma, \rho \rangle}
  \and
  \inferrule{\sigma(x) \neq c}{\langle \textbf{while } x=c\ \textbf{do } S\ \textbf{od}, \sigma, \rho \rangle 
  \to \langle \downarrow, \sigma, \rho \rangle}

    
  \end{mathpar}}
  \caption{Operational semantics for Qiskit programs}
  \Description{Operational semantics for Qiskit programs}
  \label{fig:opsem}
\end{figure}}

{\vskip 3pt}

The \textbf{classical control flow} of quantum programs has been increasingly emphasized in the study of quantum programming~\cite{fengClassical,DENG202273,Selinger2006Towards,Chadha2006Reasoning} and sequential quantum circuit design~\cite{sequential,DDequivalence2022}, and its necessity for the development of more sophisticated quantum algorithms has been widely accepted. It should be noted that although we implement our tool on the Qiskit platform due to its industrial popularity and convenient toolchain support, the model checking framework proposed in this paper is not limited to Qiskit or the current syntax. In principle, it can be adapted to other quantum programming languages that support the control-flow constructs discussed in this section. On the other hand, how to transpile and deploy control-flow operations on physical quantum hardware remains an active topic in industry, and the corresponding software ecosystem is still evolving~\footnote{\href{https://quantum.cloud.ibm.com/docs/en/guides/classical-feedforward-and-control-flow}{https://quantum.cloud.ibm.com/docs/en/guides/classical-feedforward-and-control-flow}}.

\section{System Modeling and Temporal Logic}\label{sec3}
In quantum model checking, a fundamental challenge in system modeling and the definition of temporal logics for quantum programs arises from the fact that a finite control flow may correspond to an uncountably infinite set of quantum states. Several approaches have been proposed in the literature to address this difficulty.

One line of work restricts the programming model by avoiding complex quantum properties on control structures~\cite{temporal2009,do2024symbolic}, focusing only on unitary transformations and simple measurements. Another line of research develops graph-theoretic algorithms over quantum Markov chains~\cite{guan2024measurement}. While mathematically elegant, such models are often highly abstract and difficult to relate directly to concrete program instances. Both~\cite{feng2013model} and~\cite{ying2021model} adopt control-flow-based models for representing systems. However,~\cite{feng2013model} abstracts away the structural properties of quantum states and focuses solely on temporal properties defined over super-operators. In contrast,~\cite{ying2021model} does not address the issue of infinite state spaces, and the resulting framework is not compatible with classical variables.

The definitions of temporal logics also vary across existing work. In~\cite{temporal2009,do2024symbolic}, quantum properties are specified through numerical verification of amplitudes. The logic proposed in~\cite{feng2013model} does not capture properties of quantum states themselves, but instead focuses primarily on classical properties. In contrast,~\cite{ying2021model} introduces Birkhoff-von Neumann quantum logic into model checking, where the satisfaction of quantum properties is determined through subspace inclusion. This idea is adopted and further developed in our work.

Motivated by these limitations, we develop a new framework in the following sections. Our approach addresses the infinite quantum state space induced by program execution, while supporting the verification of both classical variables and quantum state properties. Furthermore, the use of subspace-based quantum logic provides a higher level of abstraction for specifying quantum properties, which helps improve the efficiency of verification.

\subsection{Quantum-classical Transition System}
Transition systems are the standard system model in model checking~\cite{baier2008principles}. 
However, quantum programs involve a continuum of quantum states and interact with classical control flow, which cannot be directly captured by classical transition systems. 
To address this, we introduce the notion of a quantum-classical transition system, where the system behavior is described by a finite set of locations, classical atomic propositions, and quantum transition relations.

\begin{definition}[Quantum-classical transition system]
    A quantum-classical transition system is a 7-tuple 
        $\mathcal{M} = (L, \mathcal{H}, \mathcal{T}, l_0, \text{AP}, f, \tilde{f}),$ where: 
    \begin{enumerate}
        \item $L$ is a finite set of locations, $\mathcal{H}$ is a finite dimensional Hilbert space, and $l_0\in L$ is a designated location, called the initial location.
        \item $\mathcal{T}\subseteq L\times \mathcal{SO}(\mathcal{H}) \times L$ is the quantum transition relation between locations, where $\mathcal{SO}(\mathcal{H})$ is the set of all super-operators (i.e. quantum operations or quantum channels, see Appendix~\ref{apd:prelim}) on $\mathcal{H}$. Each $\tau\in \mathcal{T}$ is a triple $\langle l, \mathcal{E}, l'\rangle$, which can be written as $\tau = l \stackrel{\mathcal{E}}\rightarrow l'$ where $l,l'\in L$ are the pre- and post-locations of $\tau$. 
        To comply with the rules of quantum mechanics, the transition relation should satisfy the condition: $\sum\{\!|\mathrm{tr}[\mathcal{E}(\rho)]:l \stackrel{\mathcal{E}}\rightarrow l' \in \mathcal{T}|\!\} \leq 1$. Here, $\{\!|\cdot|\!\}$ stands for a multi-set. 
        \item $\text{AP}$ is a finite set of classical atomic propositions, which represent Boolean-valued properties evaluated at locations.
        \item $f:L\rightarrow 2^{AP}$ assigns every location with a set of classical atomic propositions, and $\tilde{f}: L\rightarrow \mathrm{Fin}(\mathcal{P}(\mathcal{H}))$ assigns every location with a finite set of subspaces of $\mathcal{H}$. Here, $\mathrm{Fin}(X)$ denotes the set of finite subsets of $X$, and $\mathcal{P}(\mathcal{H})$ is the set of all  subspaces of $\mathcal{H}$.
    \end{enumerate}
\end{definition}

For each node $l$, $\text{succ}(l)$ and $\text{pre}(l)$ is the set of successor nodes and predecessor nodes of $l$. Following the convention of classical transition system~\cite{baier2008principles}, we apply a trivial transition $l \stackrel{\mathcal{I}}\rightarrow l$ to the locations which do not have a successor node, where $\mathcal{I}$ denotes the identity operator on $\mathcal{H}$. 

\begin{example}[Transition system of RUS program]\label{trans-prog}
Consider the Qiskit program in  Example~\ref{eg:rus1}. 
We can easily acquire a quantum-classical transition system (shown in Fig.~\ref{fig:rus_buggy_qcts}) from its control flow graph, where $\mathcal{H}$ is the $4$-dimensional Hilbert space of qubits $q_0$ and $q_1$, $L=\{l_i\}_{i=0}^{14}$, $AP = \{c\}\cup\{l_i\}_{i=0}^{14}$. For the labeling functions, we first set $\tilde{f}(l_0) = \{\mathrm{span}(\{|000\rangle\})\}$ as the initialization of quantum states. In addition, the classical proposition $c$ is included in the labeling of locations $l_{12}$, $l_{13}$, and $l_{14}$ (depicted as colored nodes), indicating that the classical register $c$ has value $1$ at these locations. For the purpose of subsequent verification, we further define $l_i\in f(l_{i})$. For all locations except $l_0$, $\tilde{f}(l)=\emptyset$.
In QisMC, the construction of the quantum-classical transition system follows the control flow graph of the program. The detailed construction process will be discussed in Section~\ref{design}.
\end{example}

\subsection{Quantum Computation Tree Logic}\label{sec:qctl-def}

In this subsection, we introduce a temporal logic for specifying properties of quantum programs. 
Our formulation adopts the notion of atomic propositions from Birkhoff-von Neumann quantum logic~\cite{birkhoff1975logic}, which has also been explored in the model checking theories of~\cite{ying2021model}. 
To address two issues that are not handled in~\cite{ying2021model}---namely, reasoning over uncountably many quantum states and compatibility with classical propositions---we define a hierarchical fusion of temporal logics. 
Moreover, since quantum measurements in a quantum-classical transition system intrinsically induce branching control flow with multiple possible evolutions, we adopt a Computation Tree Logic-style structure named qCTL to capture the temporal semantics.

\textbf{Birkhoff-von Neumann logic} is a propositional logic for reasoning about static properties of quantum systems~\cite{birkhoff1975logic,Chiara2004Reasoning,ying2021model}. The propositional formulas $\psi$ are constructed from a set $AP$ of atomic propositions using  logical connectives $\neg$ (negation) and $\land$ (conjunction). For their semantics, suppose $\mathcal{H}$ is a finite-dimensional Hilbert space and $\mathcal{P}(\mathcal{H})$ is the set of subspaces of $\mathcal{H}$. It is  well-known that $(\mathcal{P}(\mathcal{H}),\cap,\lor,\bot)$ is an orthomodular lattice with partial order $\subseteq$ (inclusion), where  $\cap$ and $\bot$ stand for intersection and orthocomplement of subspaces, respectively, and the join $\lor$ is defined as the subspace spanned by the linear combinations of vectors of all its operands.
$$\bigvee_i X_i = \mathrm{span}\left(\bigcup_i X_i\right).$$
Then each atomic proposition $A$ in $AP$ is interpreted as a subspace $\llbracket A\rrbracket$ in $\mathcal{P}(\mathcal{H})$, and logical connectives $\neg$ and $\land$ are interpreted as $\bot$ and $\cap$, respectively. Therefore, the  semantics of any propositional formula $\psi$ is also a subspace of $\mathcal{H}$, write $\llbracket \psi\rrbracket$. It can be computed by applying (the semantic interpretation of) connectives to  (the interpretation of)  atomic propositions in $A$.
For example, disjunction $\lor$ can be defined as a derived connective by $\psi_1\lor \psi_2:= \neg(\neg \psi_1\land\neg \psi_2)$, and it is easy to show that $\llbracket \psi_1\lor \psi_2\rrbracket = \llbracket\psi_1\rrbracket\lor\llbracket\psi_2\rrbracket$, which is consistent with the definition of join. Moreover, satisfaction of a proposition $\psi$ by a quantum state $\rho$ is defined as follows:
$$\rho\models \psi \text{ iff } \text{supp}(\rho)\subseteq \llbracket\psi\rrbracket.$$
In the following discussion, unless otherwise stated or causing ambiguity, we abuse the use of a logical formula $\psi$ and its semantics $\llbracket \psi\rrbracket$, using the logical notation to denote a subspace.

\textbf{Syntax of qCTL}: Our quantum computation tree logic qCTL is built upon Birkhoff-von Neumann logic. For the purpose of model checking Qiskit programs, qCTL should accommodate both quantum logical propositions and Boolean logical propositions, and  have unified logical connectives and temporal operators. We constructed a hierarchical fusion of CTL with Birkhoff-von Neumann logic to support these requirements. 
  
\newcommand{\cand}{\mathbin{\&}}
\newcommand{\cneg}{\mathbin{!}}
\begin{definition}\label{qctl}
    Quantum computation tree logic has a  three-layer structure:
    \begin{itemize}
        \item Quantum propositional formulas:\ \ $\psi::= A\mid \neg \psi\mid \psi_1\land \psi_2$
        \item State formulas:\ \  $\Phi ::= \psi\mid c \mid \exists\phi\mid \forall\phi\mid \cneg\Phi\mid \Phi_1 \cand \Phi_2$
        \item Path formulas:\ \ $\phi::=O\,\Phi\mid \Phi_1\,U\,\Phi_2$
    \end{itemize}
    Here, $A$ stands for an atomic proposition of Birkhoff-von Neumann logic, and $c$ represents a classical atomic proposition. The symbols $\neg$ and $\land$ are quantum connectives as defined previously, whereas the Boolean connectives $\cneg$ and $\cand$ denote negation and conjunction in classical propositional logic. The path quantifiers $\exists, \forall$ capture the branching structure of a computation tree, expressing \textit{for at least one path} and \textit{for all paths} respectively. The temporal operators $O$ (next) and $U$ (until) appear in path formulas, whose interpretations are consistent with those of the classical CTL~\cite{baier2008principles}. 
    The  abbreviations $\square$ (always) and $\Diamond$ (eventually) are defined in the familiar way~\cite{baier2008principles}.
    
\end{definition}

\textbf{Remark}: Readers may wonder whether it is necessary for qCTL to encapsulate quantum propositional formulas as a separate layer. In particular, one might ask whether quantum propositions could be placed directly in state formulas, as in~\cite{ying2021model}, where quantum logical connectives and path quantifiers are mixed at the same syntactic level. However, such a design is problematic because the satisfaction relation for quantum propositions is not inherently Boolean. For example, given a quantum state $\rho$ and a quantum proposition $P$, the law of excluded middle does not generally hold, i.e., it is not
necessarily the case that either $\rho \models P$ or $\rho \models \neg P$. Instead, it may happen that $0 < \mathrm{tr}(P\rho) < 1$. For this reason, quantum logical connectives such as $\neg$ are restricted to the level of quantum propositions and cannot appear in front of path quantifiers. For instance, a formula of the form $\Phi = \neg \exists O A$ would imply a purely Boolean interpretation, namely that it is false that there exists a path whose next location satisfies $A$, which is incompatible with the non-Boolean nature of quantum proposition satisfaction.

\begin{example}[Temporal logic specification of RUS program]\label{eg:qctl-formula}
  For the RUS program in  Example~\ref{trans-prog}, the following formula is a valid qCTL specification: 
  \begin{align}\label{fm:rusdebug}
      \forall \square \left(l_{14}\to  \mathrm{span}\left(\frac{1}{\sqrt{3}}|001\rangle+\frac{i\sqrt{2}}{\sqrt{3}}|011\rangle\right)\right).
  \end{align}
where $\to$ denotes the classical implication. Intuitively, it means that whence the program point reaches location $ l_{14}$, the quantum variables will be in the state $\frac{1}{\sqrt{3}}|001\rangle + \frac{i\sqrt{2}}{\sqrt{3}}|011\rangle$ (up to a global phase), which actually describes the program's functionality. We will later see that model checking this formula can reveal the bug in Example~\ref{eg:rus1}.

Another example is: $\exists\Diamond\neg\mathcal{H}$ where $\mathcal{H}$ represents the entire Hilbert space. In this case, $\neg\mathcal{H}$ corresponds to the orthogonal complement $\mathcal{H}^\bot$, which is the zero-dimensional subspace. Intuitively, this specification states that there exists a location that is eventually reachable where the quantum state space collapses
to the zero subspace.
\end{example}

\textbf{Semantics of qCTL}: Given a quantum-classical transition system $\mathcal{M}$ together with its labeling functions, the semantics of a state formula $\Phi$ in qCTL is defined by the satisfaction relation $l \models \Phi$ for each location $l$.
Intuitively, the annotations in the labeling function describe quantum propositions that should hold at particular program locations. For example, if a quantum proposition $\Phi_1$ is annotated at the initial location $l_0$, then a location $l_i$ should satisfy a proposition $\Phi_2$ only if every quantum state $\rho$ satisfying $\Phi_1$ evolves from $l_0$ to $l_i$ into a state that satisfies $\Phi_2$.

However, annotations may appear at multiple locations in the program, and the quantum operations along transitions may cause these local annotations to become mutually inconsistent when propagated through the system. Therefore, a mechanism is needed to determine which quantum propositions can consistently hold at each location.

To address this issue, we introduce the notions of \emph{strongest post-condition} (sp) and \emph{weakest pre-condition} (wp). The strongest post-condition of a location $l_i$ characterizes the subspace generated by all quantum states that evolve to $l_i$ \textbf{from states satisfying the annotations at preceding locations}. Conversely, the weakest pre-condition characterizes the largest subspace of states at $l_i$ such that every state in this subspace, \textbf{when evolved along any execution path to a location} $l_k$, satisfies the annotation at $l_k$.
Under the ordering of subspace inclusion, these two constructions correspond to the least upper bound and greatest lower bound of the quantum propositions that can be consistently satisfied at a location.

\paragraph{Strongest post-condition.} 
Let function $\upsilon: \mathcal{P}(\mathcal{H})^{L}\rightarrow \mathcal{P}(\mathcal{H})^{L}$ be defined as follows: for any $X=(X_l)_{l\in L},$ we set $\upsilon(X)=(Y_l)_{l\in L}$, where:  
\begin{align}
    Y_l = X_l\lor \left(\bigvee_{k\in \text{pre}(l)}\mathcal{E}_{kl}(X_k)\right)\lor \bigvee \tilde{f}(l) \label{upsilon_sp}.
\end{align}
Here,  $\mathcal{E}_{kl}$ stands for the super-operator in the transition from location $k$ to $l$, and  $\mathcal{E}_{kl}(X_k)$ is the image of $X_k$ under $\mathcal{E}_{kl}$. By Tarski's fixed point theorem, we know that $\upsilon$ has the least fixed point, say $X^\ast=(X^\ast_l)_{l\in L}$. It can be computed by iteration as widely employed in classical model checking~\cite{McMillan1993Symbolic}. Then the strongest post-condition of location $l$ is defined as  $sp[l]=X_l^\ast$.


\paragraph{Weakest pre-condition.} 
Let function $\Upsilon: \mathcal{P}(\mathcal{H})^{L}\rightarrow \mathcal{P}(\mathcal{H})^{L}$ be defined as follows: for any $X=(X_l)_{l\in L},$ we set $\Upsilon(X) = (Y_l)_{l\in L}$, where:
\begin{align}
    Y_l = X_l\land 
    \left(\bigwedge_{k\in \text{succ}(l)}\mathcal{E}_{lk}^{-1}(X_k)\right)\land \bigwedge \tilde{f}(l) \label{upsilon_wp}.
\end{align}
Here,  $\mathcal{E}_{lk}^{-1}(X)$ is defined according to the three operation types in Fig.~\ref{fig:opsem}, 
\begin{enumerate}
    \item If $\mathcal{E}_{lk}$ is unitary, the pre-image $\mathcal{E}_{lk}^{-1}$ is the inverse of $\mathcal{E}_{lk}$.
    \item If $\mathcal{E}_{lk}$ is a projector $P\in\mathcal{P}(\mathcal{H})$, 
    $P^{-1}(X) = ( P\land  X)\lor (I-P)$, $I$ is the identity projection,
    \item If $\mathcal{E}_{lk}$ is an initialization to the state $|\psi\rangle$,
    \begin{align*}
    \mathcal{E}_{lk}^{-1}(X) = 
    \begin{cases}
    0, & \text{if }|\psi\rangle \notin  X\\
    I, &  \text{otherwise.}\\
    \end{cases}
    \end{align*}
\end{enumerate}
Especially, if $\tilde{f}(l)=\emptyset$, then the conjunction over the set $\bigwedge \tilde{f}(l) = \text{supp}({I})$. By Tarski's fixed point theorem and the duality, we know that $\Upsilon$ has the greatest fixed point, say $X^\ast = (X^\ast_l)_{l\in L}$. Then the weakest pre-condition of location $l$ is defined as $wp[l] = X^\ast_l$.


With the above preparation, we are now ready to define the semantics of qCTL. 

 \begin{definition}\label{def:qCTLsem}
    {Let $\mathcal{M} = (L,\mathcal{H},\mathcal{T},l_0, \text{AP}, f, \tilde{f})$ be a quantum-classical transition system; $l\in L$ is a location; and $\pi[i]$ is the $i$th location along the path $\pi$.}
     \begin{enumerate}
         \item For a quantum propositional formula $\psi$, $l\models \psi$ iff $sp(l)\subseteq \llbracket\psi\rrbracket\subseteq wp(l)$. 
         \item For state formula $\Phi$:
         \begin{enumerate}
             \item $l\models c$ iff $c\in f(l)$;
             \item $l\models \cneg\Phi$ iff $l\not\models\Phi$;
             \item $l\models\exists\phi$ iff $\pi\models \phi$ for some path $\pi$ starting in $l$;
             \item $l\models\forall\phi$ iff $\pi\models \phi$ for all paths $\pi$ starting in $l$;
             \item $l\models\Phi_1 \cand \Phi_2$ iff $l\models \Phi_1$ and $l\models \Phi_2$.
         \end{enumerate}
         \item For path formula $\phi$:
         \begin{enumerate}
             \item $\pi\models O\Phi$ iff $\pi[1]\models\Phi$;
             \item $\pi\models \Phi_1 U\Phi_2$ iff there exists $i>0$ such that $\pi[i]\models \Phi_2$ and $\pi[j]\models\Phi_1$ for all $0\leq j < i$.
         \end{enumerate}
     \end{enumerate}
 \end{definition}

The semantics of state formulas and path formulas in qCTL are similar to those in classical CTL and several existing quantum temporal logics~\cite{ying2021model,feng2013model,Tsubasa2023Semantic}. 
However, two major innovations, \textbf{the hierarchical structure} of quantum propositional formulas and \textbf{the $sp$/$wp$-based satisfaction semantics} address the two challenges discussed at the beginning of this section. In particular, the uncountably many runtime quantum states are abstracted by finite-dimensional subspaces (captured by $sp$ and $wp$). 
Moreover, the hierarchical design of qCTL allows classical variables and quantum propositions to be handled in a unified manner. 
In the following sections, we present the model checking algorithm and the corresponding reasoning procedure over subspace logic.


\section{QCTL Model Checking}\label{sec4}

In this section, we develop our algorithm for model checking Qiskit programs against their temporal properties specified in qCTL defined in the last section. 
Formally, our quantum model checking problem can be stated as follows: given a qCTL formula $\Phi$ and a quantum-classical transition system $\mathcal{M}$, does $\Phi$ hold for (the initial location $l_0$ of) $\mathcal{M}$; that is,
$$\mathcal{M}\models \Phi ?$$ 
As previously defined, the two bounding constraints $wp$ and $sp$ are employed to restrict the satisfiable quantum propositions within the locations, which avoids the problem of an infinite collection of quantum states induced by while-loops in the program.
Based on this basic idea, our qCTL model checking algorithm is organized into the following phases:
\begin{enumerate}
    \item Bidirectional fixed-point iteration of quantum operations computing $wp$ and $sp$;
    \item Checking the satisfiability of each quantum proposition.
    \item Reducing qCTL model checking problem to a classical CTL model checking problem.
\end{enumerate}
Since phase (2) was already presented in the last section, this section focuses on phases (1) and (3). Phase (3) is particularly significant, as it allows our model checker QisMC to leverage efficient classical model checking algorithms and tools as back-end solvers.

\subsection{Bidirectional Quantum Fixed-point Iteration}\label{sec:bi-fixed}
In practice, the predicates $wp$ and $sp$ are computed through a bidirectional fixed-point iteration over the control-flow graph of the program. Intuitively, the algorithm propagates quantum propositions backward and forward along transitions until a global consistency
condition is reached. Both computations follow a BFS-style iteration over the transition graph until a fixed point is reached.

Algorithm~\ref{alg:wp} computes the weakest pre-conditions for all program locations by a backward fixed-point iteration. 
Each predicate $wp[i]$ is initialized with the conjunction of the annotations at location $l_i$. 
The queue is initialized with all annotated locations and terminal locations, which serve as boundary conditions for the propagation.
During the iteration, whenever a location $l_x$ is dequeued, the predicates of its predecessors are refined. 
For each $i \in \text{pre}(l_x)$, the algorithm computes the intersection of $wp[i]$ with the pre-image of $wp[x]$ under the corresponding transition super-operator. 
If the dimension of $wp[i]$ decreases, the updated location is reinserted into the queue to further propagate the refinement. 
Otherwise, a location is propagated only once to ensure that all predecessors are explored.
The iteration terminates when no predicate can be further refined, yielding the weakest pre-condition at each location.

The computation of strongest post-conditions follows a dual procedure. Instead of backward propagation, the algorithm performs a forward propagation starting from the initial location and annotated locations. During the iteration, image operations replace pre-image operations, and disjunction replaces conjunction when updating predicates. Together, $wp[i]$
and $sp[i]$ characterize the maximal set of quantum propositions that can consistently hold at each location. Proposition~\ref{prop:wp} states the consistency between the algorithm and the definition. 
In Section~\ref{design}, we will present the implementation details of the algorithm, including the realization of the quantum reasoning procedures.

\begin{proposition}\label{prop:wp}
    For each location $l\in L$, Algorithm~\ref{alg:wp} computes the weakest pre-condition $wp[l]$.
    Dually, we have an algorithm for computing the strongest post-condition $sp[l]$ (omitted for brevity).
\end{proposition}
    

{\begin{algorithm}[htbp]
\footnotesize
    \renewcommand{\algorithmicrequire}{\textbf{Input:}}
    \renewcommand{\algorithmicensure}{\textbf{Output:}}
    \caption{Computation of weakest pre-conditions}
    \label{alg:wp}
    \begin{algorithmic}[1]
    \Require A quantum-classical transition system $\mathcal{M} = (L, \mathcal{H}, \mathcal{T}, l_0, AP, f, \tilde{f})$
    \Ensure $wp[i]$ for all $l_i \in L$
    \State $wp \gets$ array of length $|L|$ initialized as $\bigwedge \tilde{f}(l_i)$
    \Comment{Initialization}
    \State $Q \gets$ empty queue
    \State $Seen\gets \emptyset$
    \For{each location $l_i$}
    \Comment{Starting locations of BFS}
        \If {$\tilde{f}(l_i)\neq \emptyset$ \textbf{ or } $\text{succ}(l_i)=\{l_i\}$}
            \State enqueue($Q,i$)
            \State add $i$ to $Seen$
        \EndIf
    \EndFor
    \While{$Q$ not empty}
    \Comment{Fixed-point iterations}
        \State $x \gets$ dequeue($Q$)
        \For{each $i \in \text{pre}(l_x)$}
            \State $\mathrm{new\_wp} \gets wp[i] \land \mathcal{E}_{ix}^{-1}(wp[x])$
            \If{$dim(\mathrm{new\_wp})< dim(wp[i])$}
                \State $wp[i]\gets \mathrm{new\_wp}$
                \State enqueue($Q,i$)
                \State add $i$ to $Seen$
            \ElsIf{$i \notin Seen$}
                \State enqueue($Q,i$)
                \State add $i$ to $Seen$
            \EndIf
        \EndFor
    \EndWhile
    \end{algorithmic}
\end{algorithm}}

\subsection{Reduction from qCTL Model Checking to CTL Model Checking}\label{Main-Theorem}
In Definition~\ref{def:qCTLsem}, we presented the semantic satisfaction relation of qCTL formulas. In the previous section, we further provided algorithms for determining the satisfaction of quantum propositions by computing the operators $sp$ and $wp$. A key observation underlying our approach is that, although the semantics of a quantum proposition is defined over subspaces of a Hilbert space, its satisfaction at a given location is still a Boolean property. Once the subspace associated with a proposition is determined, a configuration either satisfies it or does not.

Based on this observation, each quantum proposition occurring in a formula---namely, expressions constructed solely using quantum logical connectives---can be associated with a fresh classical atomic proposition. The satisfaction of this atomic proposition at a location is defined to coincide with the satisfaction of the corresponding quantum proposition. Consequently, after evaluating all quantum propositions, the remaining verification task can be reduced to a classical CTL model checking problem. Classical model checking algorithms and tools can then be applied to determine the satisfaction of the resulting formula.

Formally, we present the following Algorithm~\ref{alg:reduction} and Theorem~\ref{thm:reduction}.


{
\begin{algorithm}[ht]
\footnotesize
    \renewcommand{\algorithmicrequire}{\textbf{Input:}}
    \renewcommand{\algorithmicensure}{\textbf{Output:}}
    \caption{Reduction from qCTL model checking to classical CTL model checking}
    \label{alg:reduction}
    \begin{algorithmic}[1]
        \Require A quantum-classical transition system $\mathcal{M} = (L, \mathcal{H}, \mathcal{T}, l_0, \mathrm{AP}, f, \tilde{f})$, a qCTL formula $\Phi$.
        \Ensure A classical transition system $TS = TS[\Phi,\mathcal{M}]$ and a CTL formula $\varphi=\varphi[\Phi,\mathcal{M}]$
        \Statex \textbf{Step 1. Construct $\varphi=\varphi[\Phi,\mathcal{M}]$ by structural induction on $\Phi$.}
        \Statex \textit{Base cases: }
        \begin{enumerate}
            \item If $\Phi = c$, set $\varphi[\Phi,\mathcal{M}] = c$.
            \item If $\Phi = \psi$ (a quantum proposition formula), introduce a fresh classical proposition $\omega_\psi$, and set $\varphi[\Phi,\mathcal{M}] = \omega_\psi$. Let $\Omega$ be the set of $\omega_\psi$ introduced in this step.
        \end{enumerate}
        \Statex\textit{Inductive steps: }
        \begin{enumerate}
            \item If $\Phi = \cneg \Phi'$,\quad $\varphi[\Phi, \mathcal{M}] = \cneg \varphi[\Phi', \mathcal{M}]$.
            \item If $\Phi = \Phi_1\cand\Phi_2$,\quad $\varphi[\Phi, \mathcal{M}] = \varphi[\Phi_1, \mathcal{M}]\cand\varphi[\Phi_2, \mathcal{M}]$.
            \item If $\Phi = \vartriangle O\, \Phi'$, where $\vartriangle\in \{\exists, \forall\}$,\quad $\varphi[\Phi, \mathcal{M}] = \vartriangle O\, \varphi[\Phi', \mathcal{M}]$.
            \item If $\Phi = \vartriangle \Phi_1\, U\, \Phi_2$, where $\vartriangle\in \{\exists, \forall\}$,\quad $\varphi[\Phi, \mathcal{M}] = \vartriangle \varphi[\Phi_1, \mathcal{M}]\, U\, \varphi[\Phi_2, \mathcal{M}]$.
        \end{enumerate}
        \Statex\textbf{Step 2. Construct $TS[\Phi,\mathcal{M}]$.}
        \begin{enumerate}
            \item Compute $sp$ and $wp$ for $\mathcal{M}$.
            \item For each pair $(\psi,\omega_\psi)$ introduced in Step 1, define
            $$
                f'(l) = \{\omega_\psi \mid sp(l)\subseteq \llbracket\psi\rrbracket\subseteq wp(l)\}, \quad \forall l \in L.
            $$
            \item Define the combined labeling function
            $$
                g(l) = f(l) \cup f'(l).
            $$
            \item Construct
            $$
                TS = \langle L, \mathrm{dom}(\mathcal{T}), l_0, \mathrm{AP}\cup\Omega, g \rangle.
            $$
        \end{enumerate}
    \end{algorithmic}
\end{algorithm}
}

\begin{theorem}[qCTL model checking reduction]\label{thm:reduction}
    For any qCTL state formula $\Phi$ and quantum-classical transition system $\mathcal{M}$, let $(TS[\Phi,\mathcal{M}],\varphi[\Phi,\mathcal{M}])$ be the output of Algorithm~\ref{alg:reduction}. Then we have: 
    $$
    \mathcal{M}\models \Phi\quad \text{ if and only if }\quad TS[\Phi,\mathcal{M}]\models\varphi[\Phi,\mathcal{M}].
    $$
\end{theorem}
\begin{proof} See Appendix \ref{proof-4.3}\end{proof}
 
The reduction algorithm proceeds in two main steps. First, the algorithm replaces the quantum propositions in the qCTL formula with fresh classical propositions. Second, it constructs a classical transition system whose labeling encodes the satisfaction of these quantum propositions.

In Step~2, we first compute the $wp$ and $sp$ operators for each location using the algorithm described in Section~\ref{sec:bi-fixed}. Then, for each location, we evaluate the satisfiability of all quantum propositions appearing in the qCTL formula and use the results as the truth values of the corresponding classical propositions. The remaining components of the new transition system $TS$, such as the set of locations and the transition relation, are inherited directly from $\mathcal{M}$.

After reformulating the transition system and the qCTL formula into their CTL counterparts, the original qCTL model checking problem can be solved using existing classical model checking tools. This overall workflow differs significantly from previous quantum model checking approaches~\cite{feng2015qpmc,do2024symbolic}. By leveraging mature classical model checking techniques, our approach mitigates the state explosion problem and enables the verification of substantially larger benchmarks.

{\vskip 4pt}

To conclude this section, let us apply the model checking algorithm presented in this section to the example of the RUS program to illustrate its practicality. 

\begin{example}[Model checking RUS program]\label{qctl-mc}
    Let us check the transition system in Example~\ref{trans-prog} against the temporal logical  formula~\ref{fm:rusdebug} in Example~\ref{eg:qctl-formula}, thereby performing the verification process for the RUS program described in Example~\ref{eg:rus1}.

    {\vskip 4pt}
    
    To get the satisfiability, the bidirectional quantum fixed-point iterations introduced in this section are conducted. The weakest pre-condition and strongest post-condition of $l_{14}$ are:
    $$wp[14]=\mathcal{H}^{2\otimes 2\otimes 2};\quad sp[14]=\mathrm{span}(|001\rangle,|011\rangle)$$
    Obviously, $\mathrm{span}\left(\frac{1}{\sqrt{3}}|001\rangle + \frac{i\sqrt{2}}{\sqrt{3}}|011\rangle\right)\subset sp[14]\subset wp[14]$, the satisfaction does not hold. In practice, we introduce a fresh classical proposition corresponding to the quantum proposition $\mathrm{span}\left(\frac{1}{\sqrt{3}}|001\rangle+\frac{i\sqrt{2}}{\sqrt{3}}|011\rangle\right)$, rewrite the formula accordingly, and then submit the specification to a classical model checker.

    Through the verification of this property, we conclude that the program implements incorrect functionality. The reason lies in the absence of reinitialization for $q_0$ inside the loop body, which causes unexpected residual states to appear at $l_{11}$ from the second iteration onward. Consequently, the post-condition at $l_{14}$ becomes larger than expected.
\end{example}

\section{System Design and Implementation}\label{design}

The preceding sections presented the theoretical foundations of QisMC. In this section, we describe the system design and the end-to-end debugging workflow of QisMC. As illustrated in Fig.~\ref{fig:QisMCworkflow}, QisMC consists of four main modules:

\begin{itemize}
    \item \textbf{Automated model construction.} This module automatically constructs a quantum transition system from an input Qiskit program.
    \item \textbf{Symbolic quantum reasoning backend.} This module provides unified representations for quantum states, operations, and subspaces, and implements the corresponding quantum reasoning and logical operations, as well as the fixed-point iteration algorithms introduced in Section~\ref{sec:bi-fixed}.
    \item \textbf{Property specification interface.} This module provides interfaces for annotating and labeling program locations, specifying quantum propositions of interest, and composing them into qCTL specifications.
    \item \textbf{Model checking and counterexample-guided debugging.} This module performs model checking on the resulting model and specification, and translates verification results, particularly counterexamples, into diagnostic information for debugging the original program.
\end{itemize}

QisMC adopts a two-layer architecture consisting of Python and C++. Among the four modules, the symbolic quantum reasoning backend --- CFLOBDD~\cite{cflobdd} in our evaluation --- is implemented in C++ for efficiency. The remaining components, including the integration with Qiskit and the input/output interfaces of the model checking workflow, are implemented in Python for flexibility and ease of use. NuSMV is employed as an external model checking engine and is invoked from the Python layer.

\subsection{Automated Model Construction}
QisMC accepts a Qiskit \texttt{QuantumCircuit} instance as program input, which can either be constructed programmatically through the Qiskit APIs, including dynamic control-flow constructs, or imported from an OpenQASM file. QisMC then automatically translates the input program into a transition system based on its control-flow graph (CFG). In the resulting transition system, each edge corresponds to a quantum operation or a control-flow transition, while each location represents a program state during execution and is assigned a unique identifier for subsequent reference. QisMC supports Qiskit’s major control-flow constructs, including \texttt{if\_else}, \texttt{while\_loop}, \texttt{switch\_case}, and \texttt{for\_loop}, as well as individually conditioned gates.

In Qiskit programs, classical registers store measurement outcomes and can subsequently be used to determine control-flow decisions. QisMC automatically registers these classical variables as built-in classical atomic propositions. At each location of the constructed transition system, their valuations are determined by the corresponding classical program state. These propositions can then be directly referenced when specifying temporal properties involving classical control flow.

\subsection{Implementation of the Fixed-point Iteration}\label{sec:implementation}
To implement the proposed fixed-point iteration, we first need an efficient representation of quantum states and quantum operators, together with basic operations on them which include the intersection, union, and complement of subspaces, as well as the computation of the image of a quantum state under a given operator. 

Due to the nature of quantum computation, representing arbitrary quantum states on a classical computer generally requires resources exponential in the number of qubits. To mitigate this challenge, we adopt decision diagrams as a symbolic representation for quantum states and operators. In this work, we employ the CFLOBDD~\cite{cflobdd} as the reasoning backend of QisMC which has previously demonstrated strong symbolic representation capabilities in related tasks such as quantum reachability analysis~\cite{dai2024qreach}.

The acronym CFL in CFLOBDD stands for \emph{context-free language}. Intuitively, CFLOBDDs exploit hierarchical reuse of structurally similar parts in the original matrix representation, allowing them to share repeated substructures in a manner analogous to procedure calls. This design enables significant compression for structured inputs. In particular, for certain classes of matrices, CFLOBDDs can be exponentially more succinct than traditional BDD- or ADD-based representations~\cite{robdd,algebraicDD}, and doubly exponentially more succinct than explicit matrix representations~\cite{cflobdd,wcflobdd}.

As a quantum simulation tool, CFLOBDD supports different kinds of vector and matrix calculations. Building on this fundamental capability, we aim to reduce all quantum reasoning operations involved in Algorithm~\ref{alg:wp} to vector and matrix computations, specifically: scalar multiplications, vector inner products, matrix-vector multiplications, and tensor products of matrices or vectors. A technique similar to that used in~\cite{dai2024qreach} is adopted, namely, avoiding matrix-matrix multiplications throughout the computation process in order to prevent excessive time and space overhead in decision diagram representations. 

The implementation of QisMC requires two basic techniques: 

(1) \textbf{Image Computation}: For any subspaces $X$ of a Hilbert space that occur during the execution of the model checking algorithm, once it serves as an annotation, and continues through the fixed-point iteration and subsequent satisfaction checking, it is represented by a set of orthonormal basis vectors. Thus, it is a basic step in the model checking to compute the image of such a subspace under a super-operator $\mathcal{E}$. This issue can be solved employing Proposition~\ref{prop:image} from \cite[pp.~210-211]{YINGFOUNDATION}. 
\begin{proposition}\label{prop:image}
    If $\mathcal{E}$ has the Kraus operator-sum representation $\mathcal{E} = \sum_{i\in I}E_i\circ E_i^\dagger$, then
    $$\mathcal{E}(X) = \text{supp}\{E_i|\psi\rangle: i\in I \text{ and } |\psi\rangle\in X\}$$
\end{proposition}

(2) \textbf{Intersection Computation}: In the model checking, we also often need to compute the intersection of two subspaces $V_1$ and $V_2$ of a Hilbert space (or, from a logical perspective, the conjunction of $V_1$ and $V_2$). A solution to this issue is presented as Algorithm~\ref{alg:conj}. The subroutine $\textsc{GramSchmidt}$ implements the standard Gram-Schmidt orthogonalization process, which also serves as a means of computing the union (disjunction) of subspaces. The subroutine \(\textsc{ProjectOn}\) computes the projection of a vector $|v\rangle$ onto the subspace $V_1$. It can be obtained through the following equation, suppose $V_1 = \mathrm{span}(\{|\psi_i\rangle\})$:
$$\textsc{ProjectOn}(V_1, |v\rangle) = \sum_i\langle\psi_i|v\rangle\cdot|\psi_i\rangle$$
which is a combination of CFLOBDD vector manipulations.

{\begin{algorithm}[ht]
\footnotesize
\renewcommand{\algorithmicrequire}{\textbf{Input:}}
\renewcommand{\algorithmicensure}{\textbf{Output:}}
\caption{Intersection of Two Hilbert Subspaces}
\label{alg:conj}
\begin{algorithmic}[1]
\Require Orthonormal bases $V_1, V_2$ of two Hilbert subspaces.
\Ensure A set of orthonormal bases $S$ of $V_1 \cap V_2$.
\State $S\gets$ empty set
\State $V_{\text{diff}} \gets \textsc{GramSchmidt}(V_2\cup V_1)\;\setminus\;V_2$ \Comment{Orthogonalize by the order list $(V_2,V_1)$}
\For{$|v_i\rangle\in V_{\text{diff}}$}
\State $|v_i\rangle\gets \textsc{ProjectOn}(V_1, |v_i\rangle)$ \Comment{A set of vectors that is in $V_1$}
\EndFor
\State $S \gets \textsc{GramSchmidt}(V_{\text{diff}}\cup V_1)\;\setminus\; V_{\text{diff}}$
\Comment{A set of vectors in $V_1$ and orthogonal to $V_{\text{diff}}$}
\State \Return $S$
\end{algorithmic}
\end{algorithm}}

The correctness of Algorithm~\ref{alg:conj} is guaranteed by the following theorem, whose proof is deferred to Appendix \ref{proof-intersection}:
\begin{proposition}\label{thm:conj}
    Algorithm~\ref{alg:conj} computes an orthonormal basis of the intersection $V_1 \cap V_2$ for any given orthonormal bases $V_1, V_2$ of Hilbert subspaces.
\end{proposition} 

\textbf{Complexity analysis.} Let $n=|L|$ be the number of locations and $d=\dim(\mathcal{H})$ the dimension of the Hilbert space. The cost of CFLOBDD-based operations depends on both the number of nodes in diagrams and the lookup efficiency of unique tables, which are dynamic across practical scenarios. We denote by $\tau_{mv}$ and $\tau_{vv}$ the time complexities of performing a matrix-vector multiplication and an inner product of vectors in CFLOBDD representation, respectively. In the worst case, where structural compression fails and the diagrams degenerate, we have $\tau_{mv}=\mathcal{O}(d^2)$ and $\tau_{vv}=\mathcal{O}(d)$. 

The fixed-point iteration in Algorithm~\ref{alg:wp} may take up to $\mathcal{O}(n\cdot d)$ iterations. In each iteration, the algorithm computes a constant number of conjunctions and disjunctions of subspaces, which are carried out by the Gram-Schmidt procedure. Gram-Schmidt on at most $d$ vectors requires $\mathcal{O}(d^2)$ inner products, yielding $\mathcal{O}(d^2\cdot \tau_{vv}) =\mathcal{O}(d^3)$ time in the worst case. Therefore, the overall complexity of computing weakest pre-conditions (and similarly, strongest post-conditions) is
$
\mathcal{O}(n \cdot d \cdot d^2 \cdot \tau_{vv}) = \mathcal{O}(n\cdot d^3 \cdot \tau_{vv}),$ 
which leads to a worst-case upper bound of $\mathcal{O}(n\cdot d^4)$.
We remark, however, that this upper bound is rarely encountered in practice, as CFLOBDDs often achieve substantial compression, thereby reducing the effective costs of $\tau_{mv}$ and $\tau_{vv}$ by orders of magnitude.

\subsection{Property Specification Interface}
Specifying a qCTL property in QisMC involves three main steps: identifying the program locations of interest, constructing the quantum propositions to be checked, and composing these propositions into a qCTL formula.

\textbf{Identifying the program locations.} QisMC provides three levels of interfaces for identifying and annotating program locations. First, when constructing a Qiskit program, users can directly insert a named marker at a program location by invoking \texttt{qc.mark("name")}, which is provided through our wrapper around Qiskit's \texttt{QuantumCircuit}. Second, QisMC provides built-in location selectors that identify locations with particular structural roles using predefined keywords such as \texttt{leaf} and \texttt{loop}. Finally, for fine-grained control, users can directly refer to and annotate specific locations using their unique location identifiers.

\textbf{Constructing the quantum propositions.} A quantum proposition in QisMC is represented by a subspace specified by a set of basis vectors. QisMC provides multiple ways to construct such basis vectors. In particular, users can directly specify quantum states using strings in a Dirac-notation-like syntax. For example, $|01\mathord{+}\mathord{-}\rangle$ denotes a four-qubit state in which the first two qubits are in the computational-basis states $|0\rangle$ and $|1\rangle$, while the remaining two are in the Pauli-$X$ eigenstates $|+\rangle$ and $|-\rangle$. Alternatively, users can exploit the simulation capability of the CFLOBDD backend to construct a desired state by starting from the all-zero basis state and applying a sequence of unitary gates.

\textbf{Composing qCTL specifications.} Once the relevant program locations and quantum propositions have been specified, they can be composed using the qCTL operators introduced in Section~\ref{sec:qctl-def}. According to Theorem~\ref{thm:reduction}, once the satisfaction of each quantum proposition at each location has been determined, qCTL model checking can be reduced to conventional CTL model checking. Accordingly, QisMC translates the qCTL formula, the structure of the transition system, and the satisfaction information of its atomic propositions into an input model accepted by NuSMV. The resulting CTL model-checking problem is then submitted to NuSMV to obtain the verification result.

\subsection{Counterexample-Guided Debugging}
At last, QisMC invokes NuSMV on the reduced CTL problem. It returns results of true if the properties are satisfied. Otherwise, NuSMV generates a counterexample trace over transition-system states, and QisMC translates this low-level trace back into the corresponding Qiskit program statements using the location-to-statement mapping established during CFG construction, which indicates an executable sequence of program statements leading from the initial program state to the location where the specification is violated. This reflects the methodology of using model checking for software debugging as highlighted~\cite{Clarke2009lecture}.

{\vskip 3pt}

\textbf{Summary}: Finally, the user obtains a quantum-classical transition system annotated with atomic propositions and a qCTL formula. Given these two inputs, the subsequent procedure, including the transformation into a CTL problem and the generation of model checking results, is fully automated through the invocation of the \texttt{modelChecking} method in QisMC.

The techniques introduced in this section, in particular the use of CFLOBDDs together with the algorithm proposed here, enable QisMC to go beyond the typical test cases used in previous quantum model checking work~\cite{feng2015qpmc,Gay2008QMC,do2024symbolic}, such as the BB84 protocol and quantum teleportation. As a result, QisMC can be applied to the debugging of practical quantum programs. In the next section, we analyze the experimental performance of QisMC from several perspectives.

\section{Evaluation}\label{experiments}
To demonstrate and evaluate the functionality and practicality of QisMC, we formulate the following research questions and design corresponding experiments to address them:
\begin{itemize}
    \item \textbf{RQ1.} What classes of properties can QisMC specify and debug beyond existing quantum model checkers?
    \item \textbf{RQ2.} How does QisMC scale compared with existing quantum verification tools?
\end{itemize}
To address RQ1, we first present three illustrative examples in Section~\ref{Sec-Cases} to demonstrate QisMC’s verification capabilities from complementary perspectives. Two of these examples, VQSS and the quantum Bernoulli factory, have not been explored in prior quantum program verification work in this form, and the VQSS case further demonstrates QisMC’s ability to detect a real bug in an existing research prototype. We then compare QisMC with representative quantum model checkers, including QPMC, QTC-Maude, and QMC, from the perspective of tool capabilities in Section~\ref{Sec:cap_analysis}. Finally, Section~\ref{Sec:prop_tax} summarizes the properties verified throughout the paper and develops a property taxonomy to characterize the expressive scope of qCTL.

To address RQ2, Section~\ref{ScaleSection} presents several groups of scalability experiments. These include comparisons with existing quantum model checkers~\cite{feng2015qpmc,do2024symbolic} and a quantum abstract interpretation approach~\cite{quantumAbsInter}, as well as a broader evaluation on quantum programs drawn from two benchmark suites~\cite{chen2022veriqbenchbenchmarkmultipletypes,nation_benchmarking_2025}. The latter experiments cover a wider range of program structures and problem sizes, including dynamic quantum circuits with measurements and classical control flow.

We briefly clarify the choice of the baselines used in this section. Model checking requires exhaustive state-space exploration and counterexample generation, making practical implementations particularly challenging, especially in the quantum setting with exponentially growing state spaces. Consequently, despite extensive theoretical studies, only a few practical quantum model checkers are publicly available. We therefore choose QPMC~\cite{feng2015qpmc} and QTC-Maude~\cite{do2024symbolic} as baselines, since they are, to the best of our knowledge, the only open-source and reproducible quantum model checkers. We additionally implemented translation scripts to convert OpenQASM programs into their input formats. For the remaining approaches~\cite{Gay2008QMC,ying2021model,guan2024measurement}, no publicly available implementations were available for reproduction or experimental evaluation.

All of our experiments were conducted on a Linux server running Ubuntu 22.04.5 LTS, equipped with two Intel Xeon Platinum 8358P CPUs (2.60 GHz, 128 threads in total) and 503 GiB of RAM.

\subsection{RQ1.1: Illustrative Examples}\label{Sec-Cases}

We further illustrate two capabilities that distinguish QisMC from existing quantum model checkers using three representative examples:
\begin{enumerate}
    \item Reasoning about quantum-state properties expressed as subspaces within temporal specifications; and
    \item Temporal verification of quantum while programs.
\end{enumerate}

Subspace properties naturally serve as quantum assertions for specifying properties of quantum states in quantum programs and circuits~\cite{ying2021model}. The following two examples illustrate QisMC's capability of verifying such assertions.

{\vskip 3pt}

\textbf{Grover's algorithm~\cite{Grover96}.} A well-known property of Grover's algorithm is that the quantum state remains within the two-dimensional subspace spanned by the initial state and the target state before and after every Grover iteration. This property has been adopted as a representative subspace assertion in previous verification approaches, including quantum abstract interpretation~\cite{quantumAbsInter} and quantum reachability analysis~\cite{dai2024qreach}. We therefore evaluate QisMC on the Grover benchmarks from VeriQBench~\cite{chen2022veriqbenchbenchmarkmultipletypes} and~\cite{quantumAbsInter}. Notably, the CFLOBDD~\cite{cflobdd} backend of QisMC demonstrates excellent scalability on this benchmark. A detailed scalability evaluation and experimental setup are presented in Section~\ref{ScaleSection}. Specifically, for the program of a Grover iteration with $|\psi\rangle$ as its initial state and $|\phi\rangle$ as its target state, the qCTL formula to be checked is:
$$\forall\Diamond\forall\square (\text{span}(\{|\psi\rangle, |\phi\rangle\})).$$

Quantum subspaces also naturally represent code spaces in quantum error-correcting codes (QECCs). With the growing importance of quantum error correction and fault-tolerant quantum computing, verifying such systems has become an increasingly important problem~\cite{huang2025,fang2024symbolic,Chen2025}. The following example is drawn from a quantum cryptographic protocol employing QECCs. QisMC successfully detects a real bug in this research prototype, demonstrating its practical debugging capability.

\textbf{Verifiable quantum secret sharing (VQSS)}~\cite{crepeau2002secure} is an important topic in quantum cryptography. In this setting, a {dealer} aims to distribute a quantum state among a group of participants. Similar to classical secret sharing, the protocol guarantees that any coalition containing fewer than a certain threshold of participants cannot reconstruct the original quantum state from their shares. At the same time, the protocol must tolerate adversarial behavior. A bounded fraction of participants may act as cheaters and provide incorrect information during the protocol. The goal of the dealer is to ensure that, despite such malicious interference, the honest participants can still reconstruct the correct quantum state and identify the cheating parties.

VQSS protocols can be constructed based on quantum error correction codes. In a simplified version of the protocol, the dealer first encodes the secret state together with several auxiliary uniform superposition states (referred to as verifier states) into a quantum code. The code is chosen such that the number of physical qubits equals the number of participants, and the number of potential cheaters is below the code’s error-correction capability. The dealer then distributes the shares by sending the corresponding qubit of each encoded state to each partner.

To enable decoding and the identification of cheaters, all honest partners apply CNOT gate between his shares of the secret state (as control qubits) and shares of the verifier states (as target qubits). Afterwards, the verifier states are measured and the outcomes are broadcast. With an appropriately chosen quantum code---typically a CSS Code---the resulting measurement outcomes form a correctable classical codeword, while \textbf{the secret state remains unchanged under the above quantum operations}. By detecting errors in the measurement outcomes, the cheaters are revealed.

VQSS serves as a key subroutine in several important multi-party quantum protocols like quantum Byzantine agreement~\cite{ben2005fast, Taherkhani_2018}. In~\cite{Taherkhani_2018}, the authors performed resource estimation and studied the feasibility of the VQSS protocol in a quantum repeater network. A protocol configuration based on the $\llbracket 5,1,3\rrbracket$ code was presented in the paper, as illustrated in Fig.~\ref{procedureVQSS}.

\begin{figure}[htbp]
    \centering
    \begin{subfigure}[b]{0.3\textwidth}
        \includegraphics[width=\linewidth]{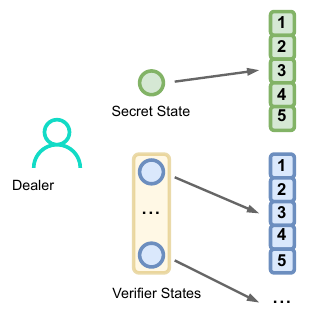}
        \caption{}
        \label{fig:pro1}
    \end{subfigure}
    \hspace{0.03\textwidth}
    \begin{subfigure}[b]{0.32\textwidth}
        \includegraphics[width=\linewidth]{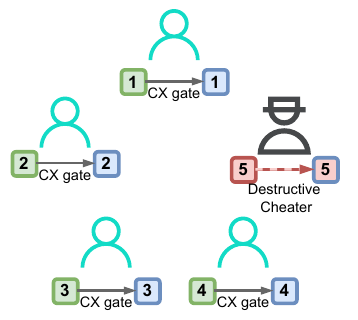}
        \caption{}
        \label{fig:pro2}
    \end{subfigure}
    \begin{subfigure}[b]{0.3\textwidth}
        \includegraphics[width=\linewidth]{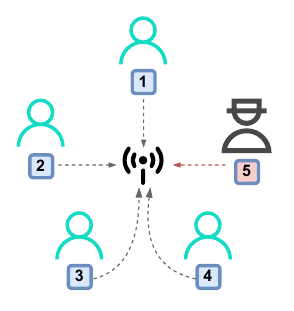}
        \caption{}
        \label{fig:pro3}
    \end{subfigure}
    \caption{The procedure of verifiable quantum secret sharing based on the $\llbracket 5,1,3\rrbracket$ code.}
    \Description{Three side-by-side images showing part of the Verifiable quantum secret sharing process}
    \label{procedureVQSS}
\end{figure}

However, we observe that this configuration does not satisfy the necessary condition described above. In particular, the secret states in the protocol are destroyed by the CNOT operations and subsequent measurements, which prevents the honest partners from correctly reconstructing the secret. To investigate this issue, we implemented the relevant part of the protocol in Qiskit and verified it using QisMC. For the classical propositions, we label the program location immediately after all quantum states are encoded as $enc$, and the final program location as $leaf$. In addition, we record the strongest post-condition computed at location $enc$ as a quantum proposition $\mathcal{C}$. The unsatisfiability of the following formulas obtained by verification indicates that the secret state is altered during the execution of the protocol.
\begin{align*}
    &\forall\square (leaf\rightarrow \mathcal{C})\\
    &\forall\square(enc\rightarrow\forall\,O\,\forall\Diamond \mathcal{C})
\end{align*}
We conclude this section with the above example drawn from a concrete research setting. Although this example does not fully showcase QisMC’s capability to verify programs with while-loops, it demonstrates the ability of \textbf{our approach to detect real bugs}, rather than artificially injected ones. This case study also suggests that QisMC and its underlying theoretical framework have the potential to be applied to a broader range of scenarios.

{\vskip 3pt}

Loops are a fundamental construct of transition systems and are essential for expressing temporal properties over potentially unbounded executions. However, support for quantum \texttt{while} programs remains limited in existing quantum model checkers; to the best of our knowledge, only QPMC explicitly includes examples involving \texttt{while} loops.  The following example features nested \texttt{while} loops and \texttt{if} statements, representing a substantially more complex control-flow structure than existing benchmarks and, to the best of our knowledge, has not previously been formulated as a quantum program for verification.

\textbf{Quantum Bernoulli factory}~\cite{dale_provable_2015} is an important algorithm that embodies provable quantum advantage. Its objective is to simulate a bias Bernoulli test $f(p)$ perfectly by using a finite number of coin flips, given access to a coin with an unknown bias $p$. The quantum version of the algorithm relies heavily on measurement-based control and repeat-until-success constructs. 
We present here an instance of the bias function $f(p)$ which cannot be implemented by classical Bernoulli factory algorithms, but can be easily achieved by a quantum algorithm, that is, 
$f(p) = 4p\cdot(1-p).$
Fig.~\ref{fig:qbfCirc} is the program and circuit graph of quantum Bernoulli factory for $p = 0.2$. The values in classical bits \textit{res} after the program stand for the results of the Bernoulli test: \texttt{01} for the head and \texttt{10} for the tail. We can verify the following temporal properties for this example:
\begin{enumerate}
    \item \emph{Liveness property.} \textit{The state of the working qubit $q1_0$ and $q1_1$ can eventually be in Bell states: $|\Psi^+\rangle=(|01\rangle+|10\rangle)/\sqrt{2},\quad |\Phi^-\rangle = (|00\rangle-|11\rangle)/\sqrt{2}$.}
    $$\exists \Diamond (\mathrm{span}(|\Psi^+\rangle)\lor \mathrm{span}(|\Phi^-\rangle))$$
    \item \emph{Safety property.} \textit{At any moment, if the value of the register `res' is \texttt{10}, it will eventually be in \texttt{10} at the end of the program.}
    $$\forall\square ((p_{10}\cand reached)\rightarrow \cneg\exists(reached\,U\,(reached\cand leaf\cand\cneg p_{10})))$$
\end{enumerate}

\begin{figure}[htbp]
    \centering
    \begin{subfigure}[b]{0.56\textwidth}
        \includegraphics[width=\linewidth]{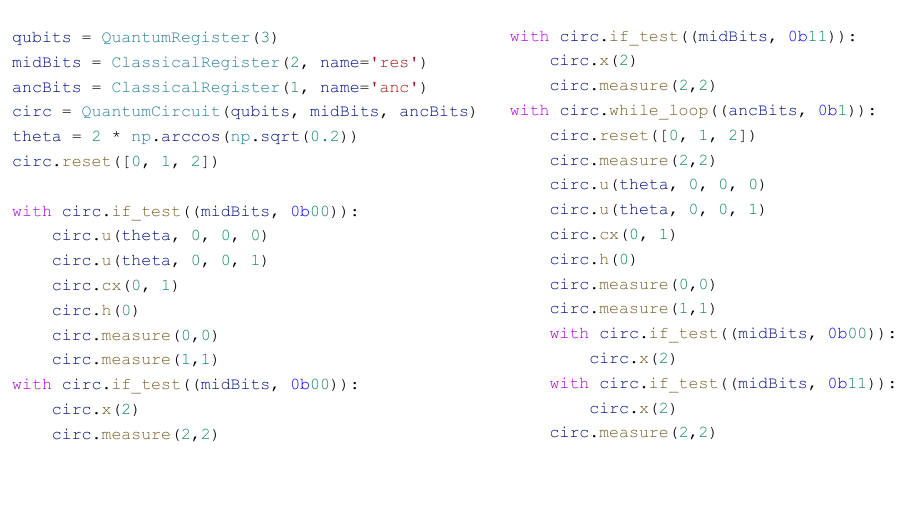}
    \end{subfigure}
    \begin{subfigure}[b]{0.4\textwidth}
        \includegraphics[width=\linewidth]{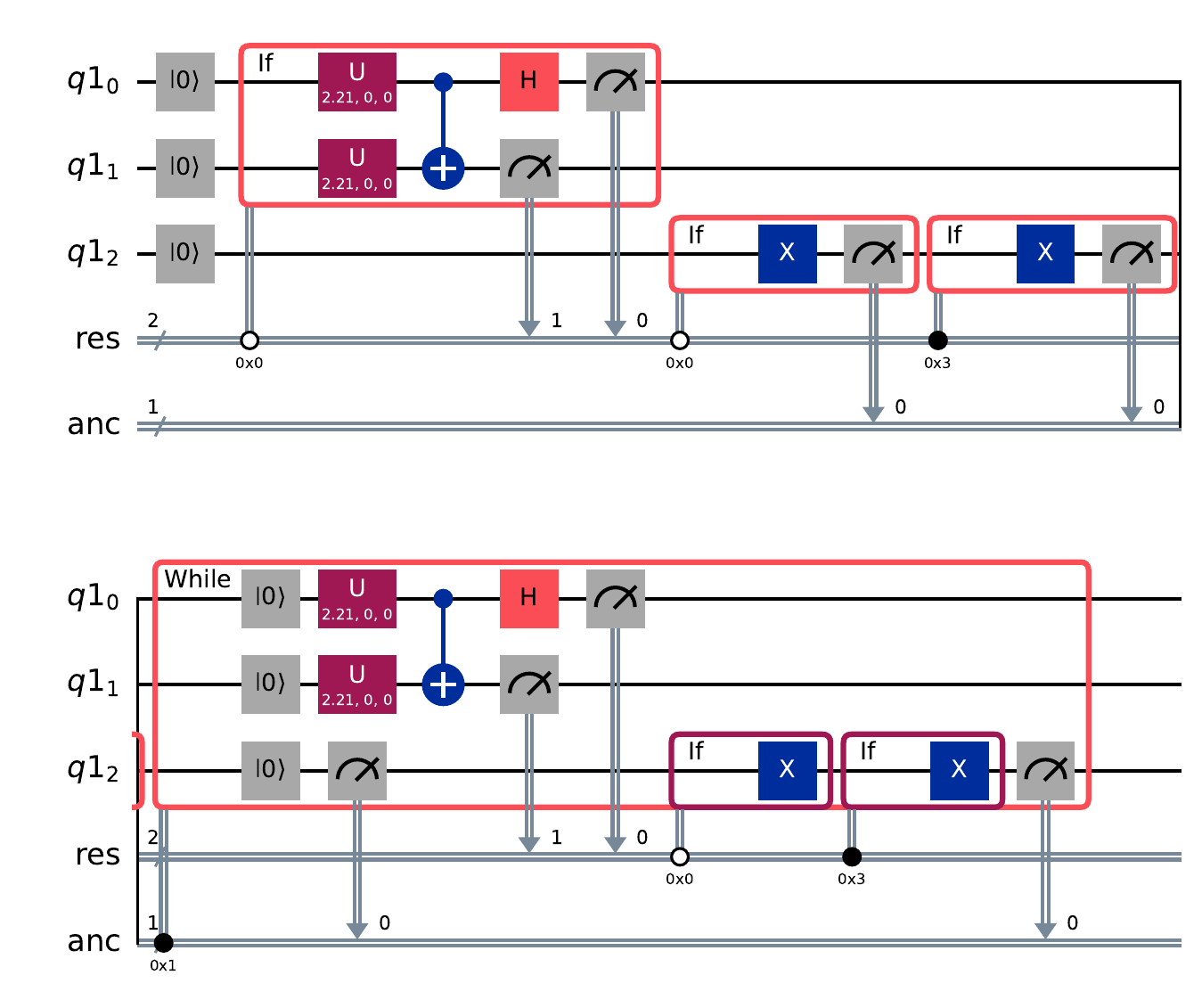}
    \end{subfigure}
    \caption{Qiskit program and circuit for quantum Bernoulli factory $f(p) = 4p\cdot(1-p)$.}
    \Description{Qiskit circuit for quantum Bernoulli factory}
    \label{fig:qbfCirc}
\end{figure}
The atomic propositions are defined as follows:
\texttt{reached} denotes locations with a positive $sp[l]$;
\texttt{p10} represents states where the classical register value is \texttt{10}; and
\texttt{leaf} identifies the leaf nodes of the transition system. 

Verification of the second property reveals an interesting fact that, during the algorithm, 
the outcome \texttt{11} on \texttt{midBits} will never occur, as the outcome \texttt{10} can never appear as an intermediate result.
To demonstrate the debugging capability of QisMC, we deliberately inject a minor fault into the program by replacing `\texttt{0b00}' with `\texttt{0b01}' in the first if-condition of the while-loop. 
In this case, QisMC not only reports a failed model checking result but also provides a clear counterexample pinpointing the source of the bug: a trace from the initial location to the terminal, consisting of 46 locations with their corresponding satisfied atomic propositions.

\subsection{RQ1.2: Capability Analysis of QisMC}\label{Sec:cap_analysis}
\newcommand{\cmark}{\ding{51}}
\newcommand{\xmark}{\ding{55}}
\begin{table}[htbp]
  \centering
  \caption{Capability comparison among quantum model checking approaches.}
  \label{tab:capability-comparison}
  \small
  \setlength{\tabcolsep}{5pt}
  \renewcommand{\arraystretch}{1.12}
  \begin{tabular}{llcccc}
    \toprule
    \textbf{Dimension}
      & \textbf{Capability}  & \textbf{QisMC}  & \textbf{QPMC}~\cite{feng2015qpmc} & \textbf{QTC-Maude}~\cite{do2024symbolic} & \textbf{QMC}~\cite{Gay2008QMC} \\
    \midrule

    \multirow{3}{*}{\textbf{Programming}}
      & Program input  & Qiskit  & PRISM & Maude  & DSL \\
      & Auto-construction & \cmark & \xmark  & \xmark & \xmark \\
      & Control-flow & General & General  & Limited & Limited \\
    \midrule

    \multirow{4}{*}{\textbf{Specification}}
      & Base Temporal logic  & CTL & PCTL & LTL  & CTL \\
      & Quantum AP & Subspaces  & \xmark & State equivalence & dEQPL \\
      & Classical AP & \cmark & \cmark & \cmark & \cmark \\
      & Probabilistic & \xmark & \cmark & \xmark & \xmark \\
    \midrule

    \multirow{2}{*}{\textbf{Verification}}
      & Symbolic & \cmark & \xmark & \cmark & \cmark \\
      & Diagnostic feedback
      & \cmark & Limited  & Limited  & \xmark \\
    \bottomrule
  \end{tabular}
\end{table}
Table~\ref{tab:capability-comparison} compares the capabilities of QisMC with three representative quantum model checkers: QPMC~\cite{feng2015qpmc}, QTC-Maude~\cite{do2024symbolic}, and QMC~\cite{Gay2008QMC}. Since no usable implementation of QMC is publicly available, its capabilities are inferred solely from the original paper~\cite{Gay2008QMC,towards2007}. The comparison is organized along three dimensions: programming, specification, and verification.

QisMC is the only tool that natively supports Qiskit programs and automatically constructs the corresponding transition systems for verification. This automation provides a unified workflow, in which quantum programs, transition systems, and the model checking algorithm are all expressed within the same programming framework, eliminating manual translation between different modeling languages. Regarding control flow, both QisMC and QPMC support quantum programs containing branches, loops, and classical feedback. QTC-Maude neither supports loops in its underlying formalism nor in its implementation, while QMC only supports final measurements.

The second dimension concerns the specification language. ``Base temporal logic'' refers to the classical temporal logic upon which each quantum temporal logic is defined. QPMC is based on probabilistic temporal logic, whereas the remaining approaches are extensions of LTL or CTL. All four tools support classical atomic propositions, but differ in their treatment of quantum atomic propositions. The expressiveness of qCTL is analyzed in detail in the following subsection.

Finally, we compare the underlying verification techniques. Here, ``symbolic'' refers to whether symbolic representations are employed for quantum states, quantum operators, or quantum reasoning, rather than explicit vector or matrix representations. QisMC employs decision diagrams for symbolic quantum reasoning. QTC-Maude instead performs symbolic deduction through rewriting rules in Maude, while QMC is restricted to the Clifford gate set, where classical simulation techniques are directly applicable~\cite{gottesman98}. Finally, among all the compared tools, only QisMC reports counterexamples at the program-statement level, providing precise debugging information that is essential for practical model checking.

\subsection{RQ1.3: Property Taxonomy of qCTL}\label{Sec:prop_tax}
In this section, we will explore the scope of qCTL in a theoretical perspective, and summarize and classify the qCTL specifications presented in the previous sections, thereby enabling a clearer comparison with existing quantum model checkers.

{\small
\begin{table*}[htbp]
    \centering
    \caption{Taxonomy of qCTL specifications}
    \label{tab:prop_taxonomy}
    \begin{tabular}{lcc}
        \toprule
        \textbf{Properties} & \textbf{qCTL Specification} & \textbf{Example} \\
        \hline
        \textbf{Static Subspaces} & $\mathrm{span}(|\Psi^+\rangle)\lor \mathrm{span}(|\Phi^-\rangle)$ & QBF \\
        \textbf{Classical Propositions Only} & $\forall\Diamond \texttt{leaf}$ & QBF / VQSS\\
        \textbf{Quantum-Derived Propositions} & $\cneg\exists(\texttt{reached}\,U\,(\texttt{reached}\cand \texttt{leaf}\cand\cneg p_{10}))$ & QBF \\
        \textbf{Quantum Invariance} & $\forall\square (\texttt{leaf}\rightarrow \mathcal{C})$ & VQSS \\
        \textbf{Quantum Reachability} & $\exists \Diamond \neg\mathcal{H}$ & RUS \\
        \textbf{Quantum Persistence} & $\forall\Diamond\forall\square (\text{span}(\{|\psi\rangle, |\phi\rangle\}))$ & Grover \\
        \textbf{Quantum Recurrence} & $\forall\square\forall\Diamond\mathcal{C}$ & VQSS \\
        \bottomrule
    \end{tabular}
\end{table*}
}

Table~\ref{tab:prop_taxonomy} summarizes the qCTL specifications used throughout our case studies and classifies them according to their expressiveness. The taxonomy ranges from static subspace assertions to temporal properties that combine classical and quantum atomic propositions. While some categories can be partially expressed by previous approaches~\cite{feng2013model,do2024symbolic}, others are unique to qCTL. As a trade-off, qCTL does not support probabilistic specifications, which remain the main advantage of QPMC~\cite{feng2015qpmc}.

Static properties correspond to quantum assertions. When only the initial location is annotated, their semantics reduce to conventional subspace assertions used in previous work~\cite{quantumAbsInter}. More generally, qCTL allows assertions to be attached to arbitrary program locations, where the consistency condition induces upper and lower bounds on admissible states. This generalization distinguishes qCTL from previous assertion-based approaches.

The remaining categories introduce temporal reasoning. The simplest form consists of temporal properties over manually defined classical atomic propositions. For example, the specification $\forall\Diamond\texttt{leaf}$ is used to detect non-terminating executions corresponding to unreachable leaf locations. The next category consists of temporal properties over \emph{quantum-derived classical propositions}. Although their atomic propositions are still Boolean, their truth values are obtained through quantum reasoning. In our quantum Bernoulli factory example, the proposition \texttt{reached} is derived from the bidirectional fixed-point computation and is semantically equivalent to the Boolean predicate $\mathrm{Dim}(sp[\mathrm{loc}])>0$. Similar specification patterns are common in previous quantum model checking approaches. For instance, QTC-Maude~\cite{do2024symbolic} employs predicates such as $Prob>0$, whose values are likewise computed by quantum analysis before being used in temporal reasoning.

The final category represents the key contribution of qCTL. These specifications uniformly combine classical atomic propositions, quantum atomic propositions, temporal connectives, and path quantifiers within a single specification language. We illustrate four representative classes of temporal properties: quantum invariance, quantum reachability, quantum persistence (eventually always), and quantum recurrence (always eventually). Previous work has considered individual instances of some of these properties (e.g., reachability~\cite{dai2024qreach}), but lacks a unified temporal logic capable of expressing them within a single framework and supporting practical debugging. Each property category is demonstrated by one of our case studies, including RUS, Grover's algorithm, and VQSS. The recurrence example is motivated by repeatedly executed VQSS protocols, a common assumption in long-running quantum error-correction and fault-tolerant quantum computing systems.

{\vskip 3pt}

\textbf{Answer to RQ1.} 
The former sections provide positive answers to \textbf{RQ1} from three perspectives. Through illustrative examples, QisMC verifies quantum programs and specifications beyond the capabilities of previous quantum model checkers, including the detection of a real bug reported in the literature. From the perspective of tool capabilities, QisMC is the only existing approach that combines native Qiskit support, automatic transition-system construction, symbolic verification, and statement-level counterexample generation. From the perspective of specification expressiveness, qCTL not only generalizes static subspace assertions, but also provides a unified specification language that seamlessly combines classical and quantum atomic propositions within temporal properties. Such unified expressiveness is absent from previous quantum model checking approaches. 

Finally, beyond its expressive power, QisMC also demonstrates superior scalability, which we evaluate in the next subsection.

\subsection{Scalability Benchmarks}\label{ScaleSection}
To evaluate the scalability of QisMC across different program types and scales, as well as its performance relative to existing approaches, we design our experiments in three stages.
\begin{enumerate}
    \item We compare the efficiency of QisMC, QTC-Maude~\cite{do2024symbolic}, and QPMC~\cite{feng2015qpmc} on a common set of benchmarks.
    \item We compare the efficiency of QisMC and~\cite{quantumAbsInter}, which proposed a quantum program assertion verifier. Both the tools used quantum subspaces as assertions on quantum states.
    \item We further evaluate QisMC on a broader collection of benchmarks from VeriQBench~\cite{chen2022veriqbenchbenchmarkmultipletypes} and Benchpress~\cite{nation_benchmarking_2025}, which include a variety of representative quantum algorithms, some featuring measurements and classical control flow.
\end{enumerate}
All benchmarks are provided as OpenQASM~2 programs and loaded using the standard Qiskit frontend, resulting in a uniform verification workflow. Some benchmarks, such as Grover’s algorithm (Section~\ref{Sec-Cases}), require specific verification procedures, whereas the remaining benchmarks follow a common workflow. Under this workflow, QisMC assumes the all-zero initial state, injects a randomly generated Pauli error into the program, and verifies whether the erroneous program satisfies the property induced by the correct output of the original program. Although this setup is not intended to model a realistic debugging scenario, it provides a simple, automated, and consistent evaluation protocol that can be applied uniformly across a diverse collection of benchmarks.

\textbf{Comparison with previous quantum model checkers.} We selected two sets of benchmarks, Grover's algorithm~\cite{Grover96} and Bernstein--Vazirani algorithm~\cite{BValg} from VeriQBench for the evaluation. These two benchmarks have a simple gate set of gates $\{X, CNOT, H, CCX\}$, which is easy to be transformed to the input form of QTC-Maude's and QPMC's.

The Grover benchmark contains a single Grover iteration with the all-one basis state as the target state. Following the subspace property introduced in Section~\ref{Sec-Cases}, the verification target is the two-dimensional subspace $\mathrm{Span}(|\mathord{+}\mathord{+}\mathord{+}01\rangle, |11101\rangle)$ (5-qubit program case as an example), where the first three qubits are work qubits and the remaining two are ancilla qubits. For the Bernstein--Vazirani (BV) benchmark, QisMC follows the automated debugging workflow described above. In all experiments, a randomly generated Pauli error is injected to evaluate the corresponding debugging capability.

For QPMC and QTC-Maude\footnote{The toolkits of QPMC and QTC-Maude can be found at \href{https://iscasmc.ios.ac.cn/tool/qmc/}{https://iscasmc.ios.ac.cn/tool/qmc/} and \href{https://doi.org/10.5281/zenodo.10783951}{https://doi.org/10.5281/zenodo.10783951}}, we evaluate only simple properties --- program termination, since differences in programming workflows and specification languages make it difficult to reproduce an equivalent debugging setup across all tools. Nevertheless, even under this simplified evaluation, both tools exhibit clear scalability limitations. Specifically, QTC-Maude requires 229\,s and more than 4.7 million Maude rewrites to verify the 17-qubit Grover benchmark, and 1515\,s with over 7.93 million rewrites for the 11-qubit BV benchmark. QPMC supports more general control flow and probabilistic temporal logic, but on these benchmarks it successfully verifies only instances with at most four qubits; larger instances terminate with runtime exceptions.

\begin{figure}[htbp]
    \centering
    \begin{subfigure}[b]{0.47\textwidth}
        \centering
        \includegraphics[width=\linewidth]{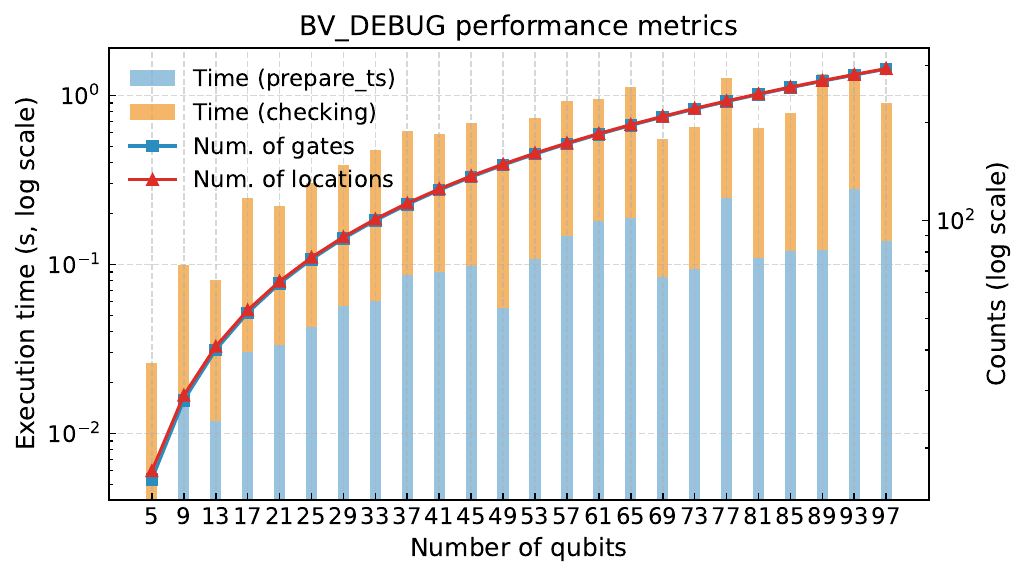}
        \label{fig:scale_bv}
    \end{subfigure}
    \hspace{0.03\textwidth}
    \begin{subfigure}[b]{0.47\textwidth}
        \centering
        \includegraphics[width=\linewidth]{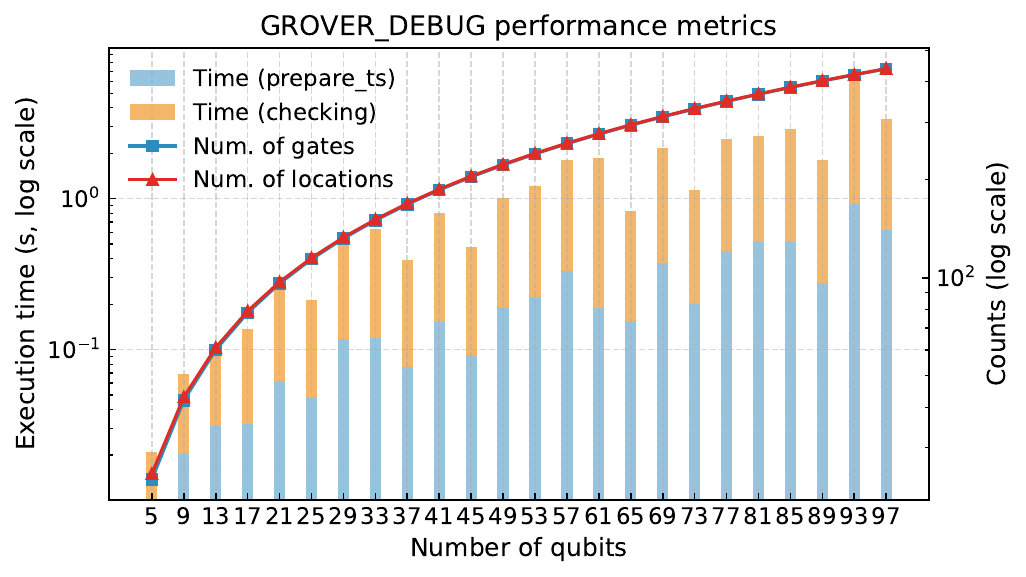}
        \label{fig:scale_grover}
    \end{subfigure}
    \caption{Performance of QisMC on Grover and BV benchmarks of VeriQBench. The execution time and number of locations are recorded. The $y$-axis is in log-scale.}
    \Description{Performance of QisMC on different quantum algorithms.}
    \label{fig:two_algorithms}
\end{figure}

The evaluation results of QisMC are recorded in Fig.~\ref{fig:two_algorithms}. It is shown that QisMC can verify quantum programs with up to 100 qubits within seconds on these two sets of benchmarks. We divide the runtime into the time of building transition systems (prepare\_ts) and the time of checking. Moreover, we have also recorded the number of gates and locations of transition systems in the graph. Because there is no measurements or control-flows in the benchmarks, the number of gates and locations are the same on each case.

Injecting a random Pauli error does not necessarily result in a specification violation. Among the 49 Grover benchmarks, 7 are verified as \texttt{True}; similarly, 28 out of the 99 BV benchmarks also satisfy their specifications after error injection.
For the BV benchmarks, whose circuits consist entirely of Clifford gates, we independently validate these cases using the Stim stabilizer simulator~\cite{gidney2021stim}. All 28 \texttt{True} instances are confirmed to preserve the correct program output. More specifically, in each case the quantum state immediately before the injected error lies in the positive eigenspace of the corresponding Pauli operator, rendering the output state unchanged. Consequently, the specification remains satisfied, and the verification results reported by QisMC are consistent with the simulator.

These results indicate that QisMC reports only errors that genuinely violate the specified property, rather than treating every injected fault as a verification failure.

\textbf{Comparison with quantum abstract interpretation~\cite{quantumAbsInter}.} The quantum abstract interpretation approach of~\cite{quantumAbsInter} aims to scale assertion verification for quantum programs and provides an implementation based on a Java matrix library. Its abstract interpretation framework is built upon quantum subspaces, which also play a central role in QisMC. Therefore, the evaluation in~\cite{quantumAbsInter}, particularly its Grover's search benchmarks, provides a relevant basis for comparison.

The two approaches nevertheless use subspaces for different purposes: \cite{quantumAbsInter} focuses on subspace-based Galois connections and the corresponding abstraction and concretization functions, whereas QisMC uses subspaces primarily as atomic propositions for logical specifications. Rather than comparing their different verification frameworks directly, we therefore focus on their common task of reasoning about quantum programs and checking subspace properties. Specifically, we investigate the advantages and limitations of QisMC's decision-diagram-based symbolic backend compared with the matrix-based implementation of~\cite{quantumAbsInter}.

{\footnotesize
\begin{table*}[htbp]
  \centering
  \caption{Results on the GHZ and Grover benchmarks of~\cite{quantumAbsInter}}
  \label{res:QAI}
  \begin{tabular}{llcccllccc}
    \toprule
    \multicolumn{5}{c}{\textbf{Grover}} &
    \multicolumn{5}{c}{\textbf{GHZ}} \\
    \cmidrule(lr){1-5} \cmidrule(lr){6-10}
    \textbf{Program} & \textbf{\#Qubits} & \textbf{\#Gates} &
    \textbf{Sat.} & \textbf{Time (s)} &
    \textbf{Program} & \textbf{\#Qubits} & \textbf{\#Gates} &
    \textbf{Sat.} & \textbf{Time (s)} \\
    \midrule
    grover-4          & 7   & 58   & True  & 0.061  &
    ghz-50            & 50  & 149  & True  & 0.542 \\
    grover-8          & 15  & 118  & True  & 0.342  &
    ghz-100           & 100 & 299  & True  & 1.079 \\
    grover-16         & 31  & 238  & True  & 3.112  &
    ghz-150           & 150 & 449  & True  & 1.639 \\
    grover-32         & 63  & 478  & --    & timeout &
    ghz-200           & 200 & 599  & True  & 2.208 \\
    grover-32-linear  & 64  & 482  & True  & 2.641  &
    ghz-250           & 250 & 749  & True  & 1.512 \\
    grover-64         & 127 & 958  & --    & timeout &
    ghz-300           & 300 & 899  & True  & 1.603 \\
    grover-64-linear  & 128 & 962  & True  & 12.179 &
                       &     &      &       &       \\
    grover-128-linear & 256 & 1922 & False & 28.533 &
                       &     &      &       &       \\
    grover-150-linear & 300 & 2252 & False & 38.056 &
                       &     &      &       &       \\
    \bottomrule
  \end{tabular}
\end{table*}
}
We evaluated QisMC on two benchmark families from~\cite{quantumAbsInter}, with all results reported in Table~\ref{res:QAI}. As a reference, we also reproduced the experiments of~\cite{quantumAbsInter}: its implementation required 101\,s and 800\,s for \texttt{grover-8} and \texttt{grover-16}, respectively, and 39\,s and 318\,s for \texttt{ghz-50} and \texttt{ghz-100}. QisMC generally demonstrates a substantial performance advantage on these benchmarks.

For the GHZ benchmarks, the verification assertion is to check whether the output state lies in the subspace $\mathrm{Span}(\{|0\cdots0\rangle, |1\cdots1\rangle\})$. QisMC successfully verifies all six instances reported in~\cite{quantumAbsInter} with high efficiency.

The Grover benchmarks use the all-zero basis state as the search target, in contrast to the all-one target used in the previous experiments. They also differ in the implementation of the oracle over ancilla qubits and in gate ordering. In particular, \cite{quantumAbsInter} provides two semantically equivalent implementations that organize the oracle in tree and linear structures, respectively. QisMC exhibits a notable sensitivity to this structural difference. For the tree-structured instances with $\#\mathrm{Qubits}\geq 63$, verification exceeds our 3200\,s timeout, whereas the corresponding linear-structured instances remain efficient. Our investigation shows that the intermediate decision diagrams generated for the tree-structured circuits grow exponentially, while this blowup does not occur for the linear variants. This result exposes a limitation of the CFLOBDD backend: its efficiency can be highly sensitive to the structural organization of the circuits.

The Grover benchmarks also reveal a numerical precision issue for the two largest programs, each containing more than 250 qubits, where QisMC fails to establish the expected subspace property. We attribute these failures to the accumulation of floating-point errors during quantum-state reasoning. We have adopted several numerical-error mitigation techniques, including higher-precision numeric types; however, without an exact symbolic representation of numerical values~\cite{Matthew2018Towards}, such errors cannot be eliminated entirely. This remains a limitation of the current implementation.

\textbf{Scalability evaluation on VeriQBench~\cite{chen2022veriqbenchbenchmarkmultipletypes} and Benchpress~\cite{nation_benchmarking_2025}.} To evaluate QisMC on a broader range of quantum algorithms, circuits, and programs, we conduct an extensive evaluation on the dynamic programs from VeriQBench~\cite{chen2022veriqbenchbenchmarkmultipletypes} and the medium-sized circuits from Benchpress~\cite{nation_benchmarking_2025}. The VeriQBench suite includes two families of distributed quantum computing (DQC) programs, implementing phase estimation~\cite{DQC_PE21} and QFT~\cite{DQC_QFT1996}, respectively. Both implementations rely on intermediate measurements and classical-bit communication, thereby introducing classical variables and control flow into the benchmarks. For comparison, we also evaluate their conventional, purely unitary counterparts.

Figure~\ref{fig:VeriQ-bench} presents the results for these four benchmark families. They reach similar maximum problem sizes, while their verification time grows exponentially with program size, exhibiting substantially worse scalability than the benchmarks evaluated before. In the clean setting, the DQC and non-DQC versions exhibit comparable verification performance, suggesting that measurements and classical control do not by themselves introduce a significant performance bottleneck for QisMC. After a random Pauli error is injected, however, the DQC versions become noticeably slower than their non-DQC counterparts. This observation suggests that the presence of control flow makes QisMC's model checking procedure more sensitive to perturbations introduced by random Pauli errors.

\begin{figure}[htbp]
    \centering
    \includegraphics[width=0.8\linewidth]{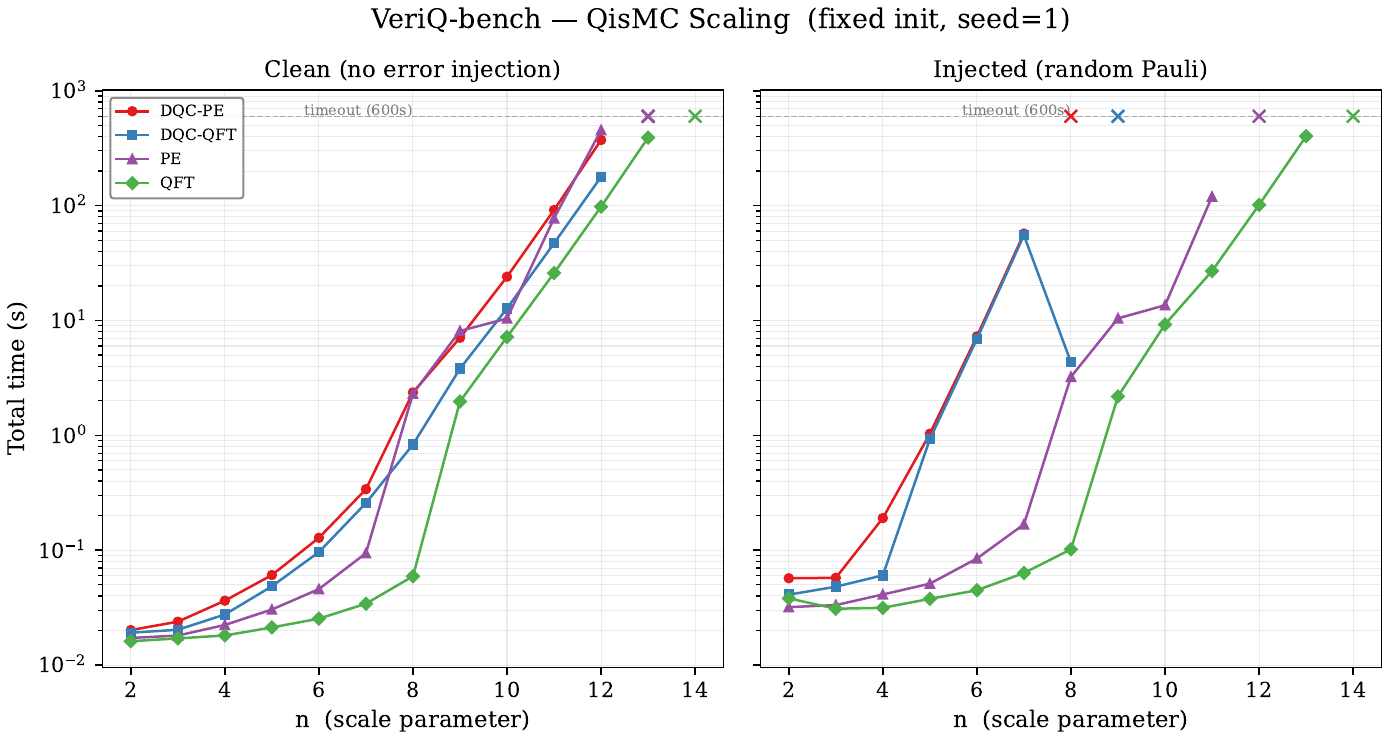}
    \caption{Scaling experimental results on typical dynamic quantum circuits and their combinational versions of VeriQBench~\cite{chen2022veriqbenchbenchmarkmultipletypes}. The debugging procedure for error-injected benchmarks is to compare the output quantum states with clean versions.}
    \label{fig:VeriQ-bench}
\end{figure}

The Benchpress-medium suite contains 40 benchmarks, covering a broader range of quantum circuits from 10 to 30 qubits. We evaluate QisMC on the complete suite using the automated workflow described above, where a random Pauli error is injected and the resulting output is checked against the property derived from the original output. Table~\ref{tab:benchpress} reports the complete results. The performance of QisMC is highly benchmark-dependent: most instances can be verified efficiently, whereas 11 instances exceed the timeout limit. Among these timeout cases, the two \texttt{bwt} instances contain more than 100,000 gates. The timeout instances from \texttt{dnn}, \texttt{ising}, \texttt{swap\_test}, \texttt{knn}, and \texttt{wstate} contain a large number of parameterized quantum gates with floating-point values, such as $R_z(0.61547971)$. Such numerical parameters reduce the sharing and compression effectiveness of decision diagrams for representing quantum states, resulting in substantially higher verification costs. The timeout of \texttt{qft\_n18} is consistent with our observations on the VeriQBench experiments, where QFT exhibits exponential growth in both decision-diagram size and verification time.

{\begin{table*}[htbp]
\footnotesize
  \caption{Results on Benchpress}
  \label{tab:benchpress}
  \begin{tabular}{lrrrrr}
\toprule
Program & \#Qubit & \#Gates & Sat & Prep. Time (s) & Veri. Time (s) \\
\hline
bigadder\_n18/bigadder\_n18\_transpiled & 18 & 340 & False & 0.347 & 0.212 \\
bv\_n14/bv\_n14 & 14 & 57 & False & 0.034 & 0.024 \\
bv\_n14/bv\_n14\_transpiled & 14 & 110 & False & 0.086 & 0.063 \\
bv\_n19/bv\_n19 & 19 & 77 & False & 0.060 & 0.046 \\
bv\_n19/bv\_n19\_transpiled & 19 & 150 & False & 0.155 & 0.117 \\
bwt\_n21/bwt\_n21 & 21 & 112833 &  & timeout & timeout \\
bwt\_n21/bwt\_n21\_transpiled & 21 & 452047 &  & timeout & timeout \\
cat\_state\_n22/cat\_state\_n22 & 22 & 46 & False & 0.049 & 0.041 \\
cat\_state\_n22/cat\_state\_n22\_transpiled & 22 & 48 & False & 0.037 & 0.026 \\
cc\_n12/cc\_n12 & 12 & 62 & False & 0.090 & 0.061 \\
cc\_n12/cc\_n12\_transpiled & 12 & 130 & False & 0.195 & 0.127 \\
dnn\_n16/dnn\_n16\_transpiled & 16 & 2848 &  & timeout & timeout \\
gcm\_n13/gcm\_h6 & 13 & 3150 & False & 16.892 & 8.509 \\
ghz\_state\_n23/ghz\_state\_n23 & 23 & 48 & False & 0.040 & 0.030 \\
ghz\_state\_n23/ghz\_state\_n23\_transpiled & 23 & 50 & False & 0.040 & 0.029 \\
ising\_n26/ising\_n26 & 26 & 307 &  & timeout & timeout \\
ising\_n26/ising\_n26\_transpiled & 26 & 203 &  & timeout & timeout \\
knn\_n25/knn\_n25 & 25 & 39 &  & timeout & timeout \\
knn\_n25/knn\_n25\_transpiled & 25 & 343 &  & timeout & timeout \\
multiplier\_n15/multiplier\_n15 & 15 & 74 & False & 0.032 & 0.023 \\
multiplier\_n15/multiplier\_n15\_transpiled & 15 & 564 & False & 0.571 & 0.309 \\
multiply\_n13/multiply\_n13 & 13 & 22 & False & 0.008 & 0.006 \\
multiply\_n13/multiply\_n13\_transpiled & 13 & 124 & False & 0.074 & 0.054 \\
qec9xz\_n17/qec9xz\_n17 & 17 & 62 & False & 0.064 & 0.046 \\
qec9xz\_n17/qec9xz\_n17\_transpiled & 17 & 104 & False & 0.115 & 0.091 \\
qf21\_n15/qf21\_n15 & 15 & 80 & False & 24.436 & 3.790 \\
qf21\_n15/qf21\_n15\_transpiled & 15 & 366 & False & 69.000 & 24.663 \\
qft\_n18/qft\_n18 & 18 & 802 &  & timeout & timeout \\
qft\_n18/qft\_n18\_transpiled & 18 & 838 &  & timeout & timeout \\
qram\_n20/qram\_n20 & 20 & 46 & False & 0.026 & 0.020 \\
qram\_n20/qram\_n20\_transpiled & 20 & 351 & False & 0.370 & 0.247 \\
sat\_n11/sat\_n11 & 11 & 96 & False & 0.107 & 0.083 \\
sat\_n11/sat\_n11\_transpiled & 11 & 736 & False & 2.229 & 1.189 \\
seca\_n11/seca\_n11 & 11 & 81 & False & 0.149 & 0.104 \\
seca\_n11/seca\_n11\_transpiled & 11 & 293 & False & 1.174 & 0.543 \\
square\_root\_n18/square\_root\_n18 & 18 & 559 & False & 2.100 & 0.978 \\
square\_root\_n18/square\_root\_n18\_transpiled & 18 & 2788 & False & 35.421 & 29.484 \\
swap\_test\_n25/swap\_test\_n25\_transpiled & 25 & 367 &  & timeout & timeout \\
wstate\_n27/wstate\_n27 & 27 & 134 & False & 0.736 & 0.372 \\
wstate\_n27/wstate\_n27\_transpiled & 27 & 237 &  & timeout & timeout \\
\bottomrule
\end{tabular}

\end{table*}}

\textbf{Answer to RQ2.} 
Our evaluation shows that QisMC is applicable to a broad range of quantum programs and achieves high efficiency and scalability on several important benchmark families, including particular implementations of Grover's search, the BV algorithm, and GHZ-state preparation. In comparison with existing quantum model checkers, QisMC also demonstrates a substantial efficiency advantage and scales to considerably larger instances on the common benchmarks evaluated in our experiments. Its performance, however, degrades on several characteristic classes of programs: algorithms inherently unfavorable to the current decision-diagram representation, such as QFT; circuits with unfavorable gate orderings or topologies, such as the tree-structured Grover benchmarks; circuits containing many floating-point gate parameters; and very large circuits where accumulated numerical errors affect verification. Taken together, these results demonstrate the scalability advantage of QisMC over existing quantum model checkers, while also revealing that the performance of its CFLOBDD-based backend remains highly dependent on program structure and numerical characteristics.

\paragraph{Limitations and future work}
QisMC currently has several limitations that motivate future work. First, qCTL does not support probabilistic properties due to its current subspace-based semantics for quantum-state properties. Although QisMC demonstrates better scalability than QPMC on the common benchmark programs evaluated in Section~\ref{ScaleSection}, extending qCTL and QisMC with probabilistic reasoning remains an important direction. Second, our experiments show that the CFLOBDD backend can be sensitive to circuit topology and parameterized gates. Moreover, its reliance on high-precision floating-point arithmetic can lead to accumulated numerical errors for very large programs. In future work, we plan to explore exact algebraic representations for important classes of quantum programs, such as Clifford+$T$ circuits, to improve both scalability and numerical robustness.

\section{Conclusion}\label{conclusion}
We introduced QisMC, the first quantum model checker designed for debugging Qiskit programs. On the theoretical side, we proposed qCTL, a quantum temporal logic grounded in Birkhoff-von Neumann logic, and defined its satisfaction semantics over quantum-classical transition systems that model the behaviors of quantum programs. We then developed a bidirectional fixed-point algorithm for computing the satisfaction of quantum propositions. By reasoning about subspace closures rather than explicitly enumerating quantum states, the algorithm avoids direct exploration of the continuous quantum state space. Based on these results, we further established a reduction from qCTL model checking to conventional CTL model checking.

On the practical side, QisMC implements an end-to-end counterexample-guided debugging workflow, extending the classical model-checking paradigm of debugging~\cite{Clarke2009lecture} to quantum programs. Our case studies and property taxonomy demonstrate that qCTL supports a broad spectrum of specifications, ranging from static subspace assertions to temporal properties that uniformly combine classical and quantum atomic propositions. Through extensive experiments, QisMC demonstrates practical advantages over existing quantum model checkers in automated program modeling, specification, diagnostic feedback, and scalability, while our evaluation also identifies limitations of the current symbolic backend. In future work, we plan to extend QisMC with probabilistic property verification and explore more scalable and numerically robust symbolic representations for quantum reasoning.


\newpage

\bibliographystyle{ACM-Reference-Format}
\bibliography{samples/reference}

\newpage
\appendix
\section{Preliminaries}\label{apd:prelim}
In this section, we will give a brief introduction to quantum computation~\cite{nielsen2010quantum}. The interested readers may refer to~\cite{ying2021model} for more information of the quantum model checking theory.

\subsection{Quantum State and Quantum Gate.} The basic element of quantum computation is a qubit. Unlike its classical counterpart, a qubit can be represented as a superposition of computational basis states $|0\rangle$ and $|1\rangle$ with complex amplitudes and a modulus of 1, or as a vector in the two-dimensional complex space $\mathbb{C}^2$: $$|\phi\rangle = \alpha|0\rangle+\beta|1\rangle = \begin{bmatrix}
    \alpha \\
    \beta
\end{bmatrix}$$
where $|\alpha|^2+|\beta|^2 = 1$. In quantum computation, a pure quantum state consists of $n$ qubits, which is the tensor product of all its qubits' vector spaces, that is, a $2^n$ dimensional complex space $(\mathbb{C}^2)^n$. A quantum state of $n$ qubits can be written in Dirac notation as a summation of computational bases:
$$|\phi\rangle = \sum_{x\in\{0,1\}^n}\alpha_x|x\rangle = \sum_{x_1,\dots,x_n}\alpha_{x_1,\dots,x_n}|x_1\cdots x_n\rangle$$
where its norm $\Vert|\phi\rangle\Vert = \sqrt{\sum_x|\alpha_x|^2} = 1$. The quantum evolution of a closed system in a pure state can be mathematically characterized as a unitary matrix in a $2^n$-dimensional complex space. In practical quantum computing, unitary quantum gates acting on localized (one or two) qubits are typically used to construct the desired quantum evolution, and their statistical data are obtained through quantum measurements. 

\subsection{Mixed State and Quantum Super-operator.} Not all quantum operators are unitary, such as the aforementioned quantum measurement, quantum noise, and qubit resetting. These more general quantum evolution operators for non-closed systems (called quantum super-operators) are part of the complete positive trace-preserving operator (CPTP). Quantum states under the influence of these operators may become mixed states. A mixed state can be mathematically expressed as an ensemble of pure states. A CPTP operator can be expressed in \textit{Kraus operator-sum form}. For example, a mixed state $\rho$ on $n$ qubits can be illustrated by an ensemble of pure states $\{p_i,|\psi\rangle\}$, written in the density matrix form $\rho = \sum_ip_i|\psi_i\rangle\langle\psi_i|$. A quantum measurement without any result information acting on the state can be written in Kraus operator: $\mathcal{E}(\rho) = \sum_i M_i\rho M_i^\dagger$, where $\sum_i E_i^\dagger E_i = I_{2^n}$. In the implementation of QisMC, CFLOBDD~\cite{cflobdd} is used to symbolically represent density matrices and super-operators by extracting support vectors and Kraus operators.

The quantum circuits and quantum programs involved in this work, that is, the objects verified by QisMC, refer to a \texttt{QuantumCircuit} object in Qiskit. After Qiskit version 1.0, the \texttt{QuantumCircuit} object has supported classical feedforward and control flow such as \texttt{if\_test}, \texttt{switch} and \texttt{while\_loop}. This provides us with a structured control flow for dynamic quantum circuits and supports complex dynamic quantum circuits and sequential quantum circuit algorithms such as quantum random walk~\cite{sequential}. Since the core of QisMC's model checking algorithm is based on the quantum-classical transition system introduced in Section~\ref{sec3} and is written in C++, future changes to the Qiskit API will only bring updates to the CFG interface layer, which is a reflection of QisMC's scalability.

\section{Proof of Theorem}
\subsection{Proof of Proposition~\ref{prop:wp}}\label{proof-4.1}
\begin{proof}
    The proof consists of three parts: (1) Algorithm~\ref{alg:wp} always terminates; (2) Soundness: The output $wp$ of Algorithm~\ref{alg:wp} satisfies $\Upsilon(wp)=wp$, where $\Upsilon$ is defined as~\ref{upsilon_wp}; (3) Completeness: For any fixed-point $Z$ of~\ref{upsilon_wp} satisfying $\Upsilon(Z)=Z$, $Z\sqsubseteq wp$, where $\sqsubseteq$ is a pointwise extension of the inclusion order $\subseteq$.
    \begin{enumerate}
        \item Since each $wp[i]$ is a subspace of the finite-dimensional Hilbert space $\mathcal{H}$, the partial order $(\mathcal{P}(\mathcal{H}), \subseteq)$ has finite height.
        In every update step, some component $wp[i]$ is replaced by a strictly smaller subspace $wp[i]\land\mathcal{E}_{ix}^{-1}(wp[x])$, hence the vector $wp$ forms a strictly descending chain with respect to the pointwise order $\sqsubseteq$.
        Because any strictly descending chain of subspaces must terminate in at most $\dim(\mathcal{H})$ steps, Algorithm~\ref{alg:wp} terminates after finitely many updates.
        \item Note that $wp$ is initialized with $\bigwedge \tilde{f}(l_i)$. Then when the algorithm terminates, there is no component can be further reduced, which means $wp[i]=wp[i]\land \mathcal{E}_{ix}^{-1}(wp[x])\land \bigwedge \tilde{f}(l_i)$ for any $i,x$, equally say, $wp[i]=wp[i]\land \bigwedge_{x\in succ(i)}\mathcal{E}_{ix}^{-1}(wp[x])\land \bigwedge \tilde{f}(l_i)$ for any $i$, i.e. $wp=\Upsilon(wp)$.
        \item Denote $wp$ after $k$-th iteration of the algorithm as $wp^k$. Let $Z$ be any fixed point of $\Upsilon$. By initialization, $wp=\bigwedge \tilde{f}(l_i)$ is the largest subspace consistent with the local constraints, thus $Z\sqsubseteq wp^0$. The update rule preserves the invariant $Z\sqsubseteq wp^k$ for all iterations $k$ since $\Upsilon$ is monotone under the pointwise inclusion order. Therefore, when the iteration terminates at $wp^\ast$, we have $Z\sqsubseteq wp^\ast$, showing that $wp^\ast$ is the greatest fixed point of $\Upsilon$.
    \end{enumerate}
    Proof completes.
\end{proof}

\subsection{Proof of Theorem~\ref{thm:reduction}}\label{proof-4.3}
\begin{proof}
    We prove $\mathcal{M}\models \Phi \Leftrightarrow TS[\Phi,\mathcal{M}]\models\varphi[\Phi,\mathcal{M}]$ by structural induction on the syntax of the qCTL state formula $\Phi$.
    \begin{enumerate}
        \item For the base case $\Phi=c$, the construction of $g(l)$ preserves all classical atomic propositions, thus the proposition is true;
        \item For the base case $\Phi=\psi$, the labelling function $f'(l)$ ensures that $l\models \psi$ in $\mathcal{M}$ iff $l\models\omega_\psi$ in $TS$;
        \item For Boolean connectives $\cneg, \cand$, since the mapping $\Phi\mapsto\varphi$ preserves Boolean connectives and the claim holds for sub-formulas, the equivalence follows immediately.
        \item For temporal operator formulas $\vartriangle O\Phi'$ and $\vartriangle \Phi_1 U \Phi_2$, where $\vartriangle\in \{\exists, \forall\}$, note that $TS$ inherits the same set of locations and the same transition graph (domain of $\mathcal T$) as $\mathcal M$. Hence there is a one-to-one correspondence between paths of $\mathcal M$ and paths of $TS$ that preserves location sequences. By the inductive hypothesis, at every location along any path the state-formula holds in $\mathcal M$ iff the reduced state-formula holds in $TS$; therefore the path formulae and path quantifiers evaluate identically in $\mathcal M$ and $TS$.
    \end{enumerate}
    Thus the equivalence holds for $\Phi$, completing the induction.


\end{proof}

\subsection{Proof of Proposition~\ref{thm:conj}}\label{proof-intersection}
\begin{proof}
For the convenience of narration, we denote $A$ for $V_1$ and $B$ for $V_2$. Let $P_B$ and $P_A$ denote the orthogonal projection operators onto $B$ and $A$, respectively, and define the linear map $T:A\to\mathcal{H}$ by $T(a)=(I-P_B)a$. Let $W:=T(A)=(I-P_B)(A)$ be the spanned space of $T(a)$ for all $a\in A$.

\emph{(1) Prove that $W$ is the spanned space of output Line 2. }
Note that $T(a)$ is exactly the Gram-Schmidt procedure of getting a new orthogonal vector from $a$, given the previous orthogonalized $B$. When Gram-Schmidt is applied to the ordered list $(V_2,V_1)$, it first orthonormalizes $V_2$, then produces additional orthonormal vectors that span the subspace of $A\lor B$ orthogonal to $B$, i.e., precisely $W$. Thus the set obtained by $\textsc{GramSchmidt}(V_2\cup V_1)\setminus V_2$ spans $W$.

\emph{(2) Get the representation of the final result $S$. }
After projecting each basis vector of $W$ back to $A$ and orthonormalizing, the algorithm obtains a subspace $D\subseteq A$ with $D=\mathrm{span}(P_A(W))$. The final step computes an orthonormal basis of the subspace
\[
S = A \cap D^\perp = A \cap \big(\mathrm{span}(P_A(W))\big)^\perp.
\]
Hence $\mathrm{span}(S)=A\cap \big(\mathrm{span}(P_A(W))\big)^\perp$.

\emph{(3) Completeness: $A\cap B \subseteq \mathrm{span}(S)$.}
Take any $x\in A\cap B$. For every $w\in W$ we have $w\perp B$, hence $\langle w,x\rangle=0$. Therefore also $\langle P_A(w),x\rangle=\langle w,x\rangle=0$, so $x\in \big(\mathrm{span}(P_A(W))\big)^\perp$. Since $x\in A$ we get $x\in A\cap \big(\mathrm{span}(P_A(W))\big)^\perp=\mathrm{span}(S)$. Thus $A\cap B\subseteq\mathrm{span}(S)$.

\emph{(4) Soundness: $\mathrm{span}(S)\subseteq A\cap B$.}
Let $y\in\mathrm{span}(S)$. Then $y\in A$ and $\langle P_A(w),y\rangle=0$ for all $w\in W$. But for any $w\in W$ there exists $a\in A$ with $w=(I-P_B)a$, hence
\[
0=\langle P_A(w),y\rangle=\langle w,y\rangle=\langle (I-P_B)a,y\rangle=\langle a,y\rangle-\langle P_B(a),y\rangle.
\]
Setting $a=y$ (which is allowed because $y\in A$) yields
\[
\|y\|^2=\langle y,y\rangle=\langle P_B(y),y\rangle=\|P_B(y)\|^2.
\]
Because $\|y\|^2=\|P_B(y)\|^2+\|(I-P_B)(y)\|^2$, the equality $\|y\|^2=\|P_B(y)\|^2$ implies $\|(I-P_B)(y)\|=0$, hence $(I-P_B)(y)=0$ and $y\in B$. Consequently $y\in A\cap B$ and thus $\mathrm{span}(S)\subseteq A\cap B$.

From (3) and (4), $\mathrm{span}(S)=A\cap B$ is proven.
\end{proof}

\end{document}